\documentclass{article}
\usepackage{het_moran_arxiv}

\usepackage{times}
\usepackage{soul}
\usepackage{url}
\usepackage[hidelinks]{hyperref}
\usepackage[utf8]{inputenc}
\usepackage[small]{caption}
\usepackage{graphicx}
\usepackage{amsmath}
\usepackage{amsthm}
\usepackage{booktabs}
\usepackage{algorithm}
\usepackage{algorithmic}
\usepackage[switch]{lineno}

\newtheorem{theorem}{Theorem}
\usepackage{bibentry}

\usepackage[dvipsnames]{xcolor}
\usepackage{docmute}
\usepackage{caption}
\usepackage{subcaption}
\usepackage{pgfplots}
\pgfplotsset{width=7cm,compat=1.3}
\usepackage{siunitx}
\usepackage{pgfplotstable}
\usepackage{pifont}
\newcommand{\cmark}{\ding{51}}%
\newcommand{\xmark}{\ding{55}}%
\usepackage{bm}
\pgfplotsset{select coords between index/.style 2 args={
    x filter/.code={
        \ifnum\coordindex<#1\fi
        \ifnum\coordindex>#2\fi
    }
}}
\usepackage{times}
\usepackage{soul}
\usepackage{url}
\usepackage[utf8]{inputenc}
\usepackage{amsmath}
\usepackage{amsthm}
\usepackage{booktabs}
\usepackage{algorithm}
\usepackage{algorithmic}
\usepackage{amsfonts}
\newtheorem{corollary}{Corollary} 
\usepackage{enumitem}
\usepackage{thmtools}
\usepackage{thm-restate}
\usepackage{parskip}
\usepackage[capitalize]{cleveref}
\usepackage{xfrac} 

\usepackage{graphicx}
\usepackage{lipsum}

\newcommand{\pk}[1]{#1} 

\DeclareMathOperator*{\argmax}{arg\,max}

\newcommand{\SeedSet}{S}
\newcommand{\SeedSetT}{T}

\newcommand{\FP}{\operatorname{fp}}

\newcommand{\NP}{\ensuremath{\operatorname{\mathbf{NP}}}}

\newcommand{\SizeConst}{k}
\newcommand{\goesto}{\rightarrow}
\newcommand{\Config}{X}
\newcommand{\ConfigY}{Y}
\newcommand{\ConfigZ}{Z}
\newcommand{\RandomConfig}{\mathcal{X}}
\newcommand{\RandomConfigY}{\mathcal{Y}}
\newcommand{\RandomConfigZ}{\mathcal{Z}}
\newcommand{\Prob}{\mathbb{P}}

\newcommand{\Expectation}{\mathbb{E}}
\newcommand{\Variance}{\operatorname{Var}}
\newcommand{\ConstantOne}{\mathbf{1}}

\newcommand{\MinResidentFitness}{\FitnessR_{\min}}
\newcommand{\MaxMutantFitness}{\FitnessM_{\max}}

\newcommand{\ActiveNodes}{\mathcal{A}}

\newcommand\numberthis{\addtocounter{equation}{1}\tag{\theequation}}
\newlist{compactenum}{enumerate}{3} 
\setlist[compactenum]{label=(\arabic*), nosep,leftmargin=*}
\crefname{compactenumi}{Item}{Items}

\newcommand{\Fitness}{f}
\newcommand{\TotalFitness}{F}

\newcommand{\FitnessR}{r}

\newcommand{\FitnessM}{m}

\newcommand{\FitnessG}{\mathcal{G}}

\newcommand{\Weight}{w}
\newcommand{\Paragraph}[1]{{\bf #1}.}

\newcommand{\MP}{\mathcal{M}}

\newcommand{\Universe}{\mathcal{U}}
\newcommand{\Sets}{\mathcal{S}}

\newcommand{\degree}{d}

\usepackage{xcolor}

\definecolor{color1}{HTML}{7EA6E0}
\definecolor{color2}{HTML}{EA6B66}
\definecolor{color3}{HTML}{97D077}
\definecolor{color4}{HTML}{CC6600}

\title{Seed Selection in the Heterogeneous Moran Process}

\author {
Petros Petsinis$^1$
\and
Andreas Pavlogiannis$^1$\and
Josef Tkadlec$^{2}$\And
Panagiotis Karras$^{3,1}$\\
\affiliations
$^1$Department of Computer Science, Aarhus University, Aarhus, Denmark\\
$^2$Computer Science Institute, Charles University, Prague, Czech
Republic \\
$^3$Department of Computer Science, University of Copenhagen, Copenhagen, Denmark\\
\emails
\{petsinis, pavlogiannis\}@cs.au.dk,
josef.tkadlec@iuuk.mff.cuni.cz,
piekarras@gmail.com
}

\begin{document}
\maketitle
\begin{abstract}
The \emph{Moran process} is a classic stochastic process that models the rise and takeover of novel traits in network-structured populations.
In biological terms, a set of \emph{mutants}, each with 
fitness~$\FitnessM\in(0,\infty)$ invade a population of \emph{residents} with fitness~$1$. 
Each agent reproduces at a rate proportional to its fitness and each offspring replaces a random network neighbor. 
The process ends when the mutants either fixate (take over the whole population) or go extinct.
The \emph{fixation probability} measures the success of the invasion.
To account for environmental heterogeneity,
we study
a generalization of the Standard process, called the \emph{Heterogeneous} Moran process.
Here, the fitness of each agent is determined both by its type (resident/mutant) and the node it occupies.
We study the natural optimization problem of \emph{seed selection}:~given a budget $k$, which~$k$ agents should initiate the mutant invasion to maximize the fixation probability?
We show that the problem is strongly inapproximable:~it is $\NP$-hard to distinguish between maximum fixation probability 0 and 1.
We then focus on \emph{mutant-biased} networks, where each node exhibits at least as large mutant fitness as resident fitness.
We show that the problem remains $\NP$-hard, but the fixation probability becomes submodular, and thus the optimization problem admits a greedy $(1-\sfrac{1}{e})$-approximation.
An experimental evaluation of the greedy algorithm along with various heuristics on real-world data sets corroborates our results.
\end{abstract}
\section{Introduction}\label{sec:introduction}

Modeling and analyzing the spread of a novel trait (e.g., a trend, meme, opinion, genetic mutation) in a population is vital to our understanding of many real-world phenomena. 
Typically, this modeling involves a \emph{network invasion process}:~nodes represent agents/spatial locations, edges represent communication/interaction between agents,  and local stochastic rules define the dynamics of trait spread from an agent to its neighbors.

Network diffusion processes raise several \emph{optimization} challenges, whereby we control elements of the process to achieve a desirable emergent effect. A well-studied problem is that of influence maximization, which calls to find a \emph{seed set} of agents initiating a peer-to-peer influence dissemination that maximizes the expected spread thereof; the problem arises in various diffusion models, such as Independent Cascade and Linear Threshold~\cite{kempe2003maximizing,domingos2001mining,mossel2007submodularity,li2011modeling,logins2020}, the Voter model~\cite{EvenDar2007,Durocher2022}, content-aware models~\cite{ivanov2017content}, models of multifaceted influence~\cite{li2019}, and geodemographic models of agent mobility~\cite{zhang2020geodemographic}.

Diffusion processes also play a key role in \emph{evolutionary dynamics}, which model the rules underpinning the sweep of novel genetic mutations in populations and the emergence of new phenotypes in ecological environments~\cite{Nowak2006}. 
A classic evolutionary process is the \emph{Moran process}~\cite{moran1958random}. 
In high level, a set of \emph{mutants}, each  with fitness~$\FitnessM\in (0,\infty)$, invade a preexisting population of \emph{residents}, each with fitness normalized to~$1$. 
Over time, each agent reproduces with rate proportional to its fitness, while the produced offspring replaces a random neighbor. 
In the long run, the new mutation either \emph{fixates} in the population (i.e., all agents become mutants) or \emph{goes extinct} (i.e., all agents remain residents). 
The probability of fixation is the main quantity of interest, especially under advantageous mutations ($\FitnessM>1$).

Network structure affects the fixation probability~\cite{Lieberman2005,Allen2017}, and may both amplify it~\cite{adlam2015amplifiers} and suppress it~\cite{Giakkoupis16,Mertzios2018}, while certain structures nearly guarantee mutant fixation~\cite{Giakkoupis16,Goldberg2019,pavlogiannis2018construction,tkadlec2021fast}.
The Moran process thus provides a simple stochastic model by which a community of communicating agents reaches consensus; one option has an advantage over another, yet its prevalence (i.e., fixation) depends on the positioning of its initial adherents (i.e., mutants) and on the network structure.

Recent work aims to make the Moran process more realistic by incorporating some form of \emph{environmental heterogeneity}~\cite{Maciejewski2014,brendborg2022fixation,melissourgos2022extension,svoboda2023coexistence}.
Here, the fitness of an agent is not only a function of its type (resident/mutant), but also of its location in space, i.e., the node that it occupies.
For example, in a biological setting, the ability to metabolise a certain sugar boosts growth more in environments where such sugar is abundant.
Similarly, in a social setting, the spread  a trait is more, or less viral depending on the local context (e.g., ads, societal predispositions).
Analogous extensions have been recently considered for the Voter evolutionary model~\cite{Anagnostopoulos2020,Becchetti2023,petsinis2022maximizing}.

In this work we generalize the classic Moran process to account for complete environmental heterogeneity, obtaining the \emph{Heterogeneous Moran process}:~for every network node $u$, a mutant (resp., resident) occupying $u$ exhibits fitness $\FitnessM(u)$ (resp., $\FitnessR(u)$) specific to that node.
We then study the  natural optimization problem of \emph{seed selection:~given a budget $k$, which $k$ nodes should initiate the mutant invasion so as to maximize the fixation probability?}
Although the seed selection problem has been studied extensively in other diffusion models,
this is the first paper to consider it in Moran models.
We obtain upper and lower bounds for the complexity of this problem in our Heterogeneous model,
which also imply analogous results to other relevant Moran models.

\begin{figure}[!t]
\centering
\includegraphics[width=0.49\textwidth]{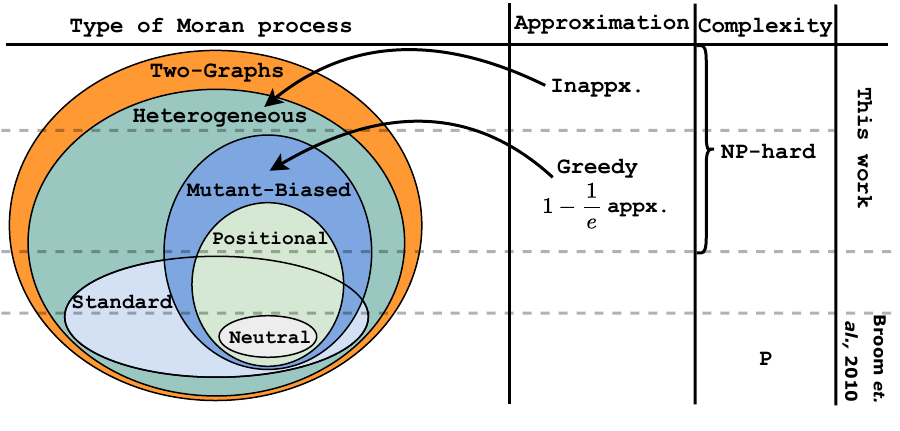}
\caption{
Moran processes (with and without environmental heterogeneity) and the complexity of seed selection.
}
\label{fig:hardness}
\end{figure}

\Paragraph{Contributions}
Our main theoretical results are as follows (see \cref{fig:hardness} for a summary in the context of Moran models).
\begin{compactenum}
\item We prove that computing the fixation probability admits a FPRAS on undirected and unweighted networks that are \emph{mutant-biased}, where $\FitnessM(u)\geq\FitnessR(u)$ for every node $u$. 
\item We show that the optimization problem is strongly inapproximable:~for any $0<\varepsilon<\sfrac{1}{2}$, it is $\NP$-hard to distinguish between maximum fixation probability $\leq \varepsilon$ and $>1-\varepsilon$.
\item We then focus on mutant-biased networks.
We show that the optimization problem remains $\NP$-hard to solve exactly, but the fixation probability becomes submodular, yielding
a greedy $(1-\sfrac{1}{e})$-approximation.
\end{compactenum}
Further, through experimentation with real-world data, we establish that the greedy algorithm outperforms standard heuristics for seed selection and uncovers high-quality seed sets with diverse data sets and problem parameters.
Due to space constraints, we relegate some proofs to the \cref{sec:app}.

\Paragraph{Technical Challenges}
Seed selection was studied recently under the Voter model, which bares some resemblance to the Moran model~\cite{Durocher2022}. However, the two models are distinct, and results in one do not transfer to the other. Some novel technical challenges we address are as follows.
\begin{compactenum}
\item 
Our NP-hardness and inapproximability proofs are fundamentally different from the NP-hardness of~\cite{Durocher2022}, and are not limited to weak selection (mutant advantage $\epsilon \goesto 0$).
\item 
Our submodularity proof is based on introducing a novel variant of the Moran process that we call the Loopy process.
This also allows us to show that the Heterogeneous Moran process
is a special case of the Two-Graph Moran process~\cite{melissourgos2022extension}, thereby extending
our hardness results to the latter.
\item 
Our model accounts for environmental heterogeneity,
while the Voter model in \cite{Durocher2022} does not. 
This complicates our proof for FPRAS.
\end{compactenum}
\section{Preliminaries}\label{sec:preliminaries}
In this section we introduce the Heterogeneous Moran process and the problem of seed selection.

\Paragraph{Population structure} 
We consider a population of agents structured as a weighted directed graph~$G =(V, E, \Weight)$, where each node~$u\in V$ stands for a single agent, each edge~$(u,v)\in E$ signifies that~$u$ influences~$v$, and~$\Weight(u,\cdot )$ is a probability distribution expressing the frequency at which~$u$ influences~$v$. $G$ is strongly connected, i.e., any two nodes are connected by a sequence of edges of non-zero weight. We call $G$ \emph{undirected} if $E$ is symmetric and $\Weight(u,\cdot )$ is uniform.

\Paragraph{Fitness graphs} 
Trait diffusion in the Heterogeneous Moran process occurs by associating each node with a type: at each moment in time, each node $u$ is either \emph{resident} or \emph{mutant}. 
Moreover, $u$ is associated with a type-dependent \emph{fitness} that represents the rate at which~$u$ influences its neighbors while being resident or mutant. We denote the respective fitness values by~$\FitnessR(u)$ and~$\FitnessM(u)$, as functions $\FitnessR, \FitnessM \colon V\to (0,\infty)$. 
We call the triplet $\FitnessG=(G, (\FitnessM, \FitnessR))$ a \emph{fitness graph}, and denote the minimum and maximum resident and mutant fitnesses in $\FitnessG$ as $\MinResidentFitness=\min_{u\in V}\FitnessR(u)$, and $\MaxMutantFitness=\max_{u\in V}\FitnessM(u)$.
We call $\FitnessG$ \emph{mutant-biased} if for all $u\in V$, we have $\FitnessM(u)\geq \FitnessR(u)$.
 
\Paragraph{The Heterogeneous Moran process} 
A \emph{configuration} is a subset of nodes $\Config\subseteq V$, representing the mutant nodes in $\FitnessG$ at some time point.
The \emph{fitness} of node $u$ in $\Config$ is defined as 
\begin{linenomath*}
\[
\Fitness_{\Config}(u) = 
\begin{cases}
\FitnessM(u),& \text{if } u\in \Config \\
\FitnessR(u),& \text{otherwise}
\end{cases}
\]
\end{linenomath*}
i.e., it is $\FitnessM(u)$ if $u$ is mutant and~$\FitnessR(u)$ if $u$ is a resident.
At time $t=0$ a seed set $\SeedSet\subseteq V$ 
specifies the nodes where mutant invasion begins.
The Heterogeneous Moran process is a 
discrete-time 
stochastic process $\RandomConfig_0, \RandomConfig_1, \dots,$ of stochastic configurations~$\RandomConfig_t \subseteq V$, where
$\RandomConfig_0 = \SeedSet$ and 
for each $t>0$,
$\RandomConfig_{t+1}$ is obtained from $\RandomConfig_{t}$ by two successive random steps:
\begin{compactenum}
\item \emph{Birth Event:} Pick a node~$u$ for reproduction with probability proportional to its fitness,
$\frac{\Fitness_{\Config}(u)}{\sum_{v\in V}\Fitness_{\Config}(v)}$.
\item \emph{Death Event:} Pick a neighbor~$v$ of~$u$ with probability~$\Weight(u,v)$ and make $v$ have the same type as $u$.
\end{compactenum}
Note that the mutant set can both grow and shrink over time.
\cref{fig:pos_moran} illustrates the process on a small example.

\begin{figure}[!t]
\centering
\includegraphics[width=0.44\textwidth]{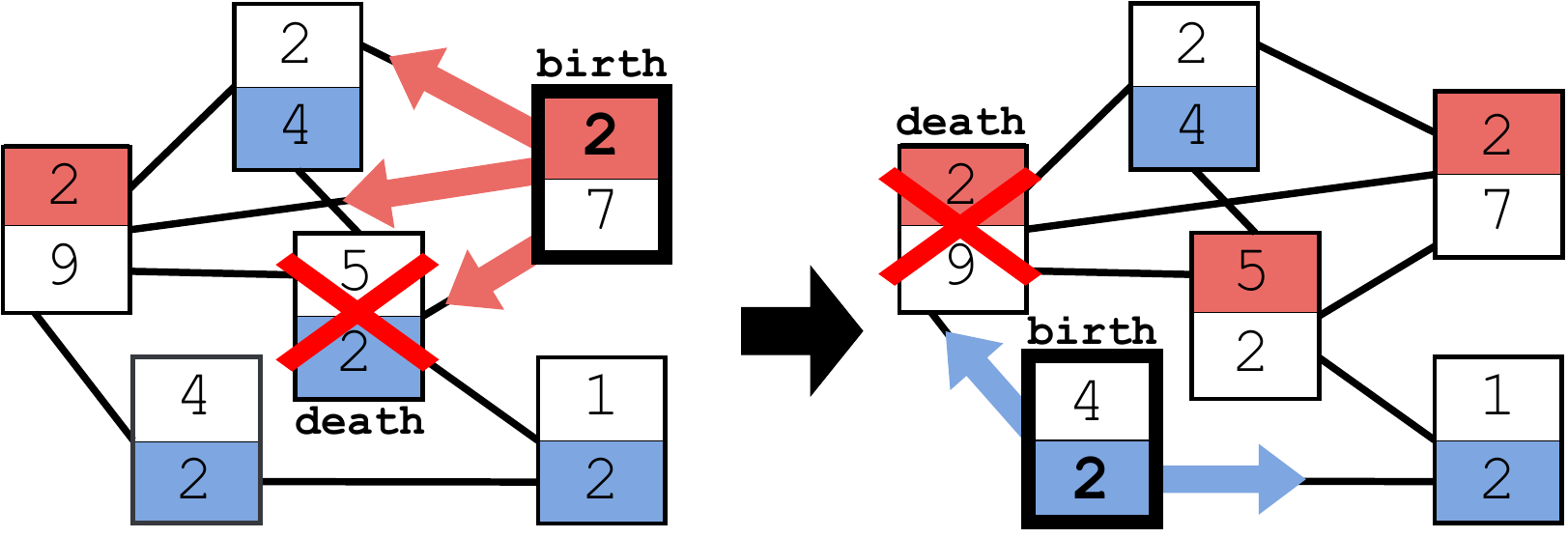}
\caption{
Two steps in the Heterogeneous Moran process; 
mutants/residents are marked in red/blue;
the numbers indicate 
type-dependent mutant/resident fitness
(top/bottom).
}
\label{fig:pos_moran}
\end{figure}

\Paragraph{Relation to other Moran processes}
We recover the Standard Moran process~\cite{moran1958random} as a special case of the Heterogeneous process with $\FitnessR(u)=1$ and $\FitnessM(u)$ constant for all $u\in V$.
The Neutral Moran process is a further special case of the Standard process, in which residents and mutants have equal fitness. 
The Positional Moran process~\cite{brendborg2022fixation} parametrizes the Standard process with an active set of nodes $\ActiveNodes$, which define the node fitness as $\Fitness_{\Config}(u)=1+\delta$ if $u\in \Config\cap\ActiveNodes$ and $\Fitness_{\Config}(u)=1$ otherwise. This is also a special case of the Heterogeneous process, with~$\FitnessR(u)=1$ and~$\FitnessM(u) = 1 + \delta$ if~$u \in \ActiveNodes$ and~$\FitnessM(u) = 1$ otherwise. The Two-Graphs Moran process~\cite{melissourgos2022extension} lets mutants and residents propagate via different, type-specific graphs $G_M$ and $G_R$, respectively, over the same set of nodes but with  different edge sets. The Two-Graphs process generalizes the Heterogeneous process, a connection formally implied by an intermediate result we derive in \cref{sec:submodularity}.

\Paragraph{Fixation probability} 
In the long run, mutants either \emph{fixate} with~$\RandomConfig_{t} = V$ or go extinct with~$\RandomConfig_{t}=\emptyset$. The \emph{fixation probability}~$\FP_{\FitnessG}(\SeedSet)$ is the probability that mutants fixate on a fitness graph~$\FitnessG = (G, (\FitnessM, \FitnessR))$ with seed set~$\SeedSet$. {The complexity of computing~$\FP_{\FitnessG}(\SeedSet)$ is an open question}, even for the Standard Moran process, in contrast to cascade spread models, for which the spread function is efficiently approximable~\cite{Svitkina2011}; as we prove in the next section, on mutant-biased, undirected fitness graphs, $\FP_{\FitnessG}(\SeedSet)$ is approximable efficiently via Monte Carlo simulations.

\Paragraph{The seed-selection problem} 
The standard optimization question in invasion processes is optimal seed placement:~\emph{given a budget $k$, which $k$ nodes $\SeedSet^*$ should initiate the mutant invasion so as to maximize the fixation probability?}
\begin{align}\label{eq:objective}
\SeedSet^*=\argmax_{\SeedSet \subseteq V, |\SeedSet|\leq \SizeConst}\FP_{\FitnessG}(\SeedSet).
\end{align}

The optimal seed depends on the graph structure,  budget $\SizeConst$, and node fitnesses.
\cref{fig:bibartite_undirected} showcases this intricate relationship, even with all residents having fitness~$1$. The optimal seed $\SeedSet^*$ may comprise
(i)~only nodes of the \emph{largest} mutant fitness ($\SizeConst=3$, left),
(ii)~nodes of both large and small mutant fitness ($\SizeConst=3$, middle), or
(iii)~only nodes of the smallest mutant fitness ($\SizeConst=3$, right).
Moreover, increasing~$\SizeConst$ may yield an optimal seed set that is not a superset of, or even disjoint to, the previous one; (left, $\SizeConst=1$ vs~$\SizeConst=3$).

\begin{figure}[!ht]
\includegraphics[width=0.475\textwidth]{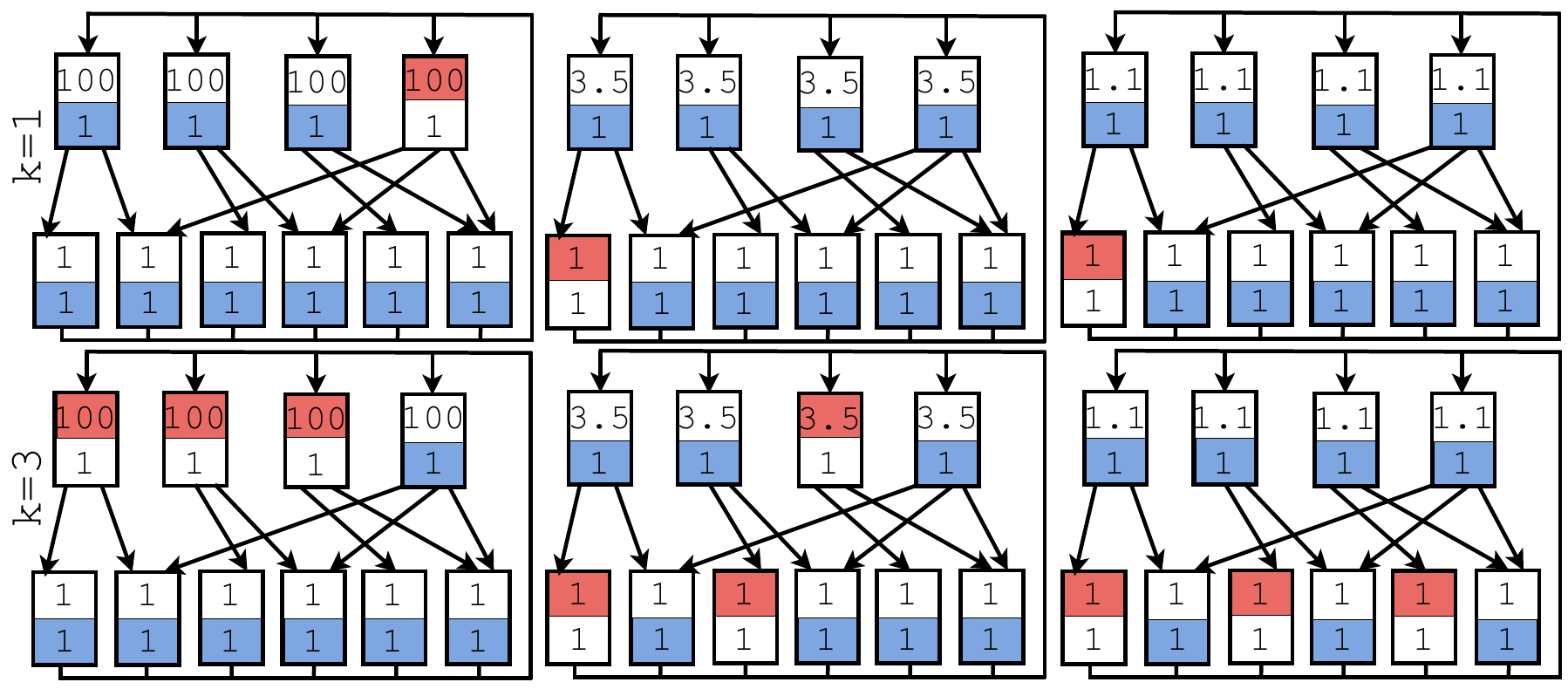}
\caption{\label{fig:bibartite_undirected}
Optimal seed set~$S^*$ (in red) while varying the mutant fitness and seed size~$\SizeConst$; all residents have fitness 1.
}
\end{figure}
\section{Computing the Fixation Probability}\label{sec:computation}

In the neutral setting ($\FitnessM(u)=\FitnessR(u),\, \forall u$), the fixation probability is linear, $\FP_\FitnessG(\SeedSet) = \sum_{u\in \SeedSet} \FP_\FitnessG(\{u\})$. 
When the graph is also undirected, $\FP_\FitnessG(\SeedSet) = \frac{\sum_{u \in \SeedSet} \sfrac1\degree(u)}{\sum_{v\in V} \sfrac1\degree(v)}$, where~$\degree(x)$ is the degree of  $x$~\cite{broom2010two}. No closed-form solution is known for the non-neutral setting.
Yet on undirected graphs the expected time until convergence is polynomial, yielding a fully polynomial-time randomized approximation scheme (FPRAS) via Monte Carlo simulations~\cite{diaz2014approximating,brendborg2022fixation}.
The next lemma generalizes this result to the Heterogeneous process on mutant-biased graphs, in sharp contrast to non-biased graphs, on which the expected time is exponential in general~\cite{svoboda2023coexistence}.

\begin{restatable}{lemma}{lemexpectedtime}\label{lem:expected_time}
Given an undirected and mutant-biased fitness graph $\FitnessG$ and a seed set $\SeedSet\subseteq V$, the expected time to convergence $T(\FitnessG,\SeedSet)$ satisfies $T(\FitnessG,\SeedSet)\leq\big(n^2\cdot\frac{\MaxMutantFitness}{\MinResidentFitness}\big)^3$.
\end{restatable}
\begin{proof}
For a configuration $\Config$, we define the potential function $\Phi(X)=\sum_{u\in \Config}\frac{\FitnessM(u)}{\degree(u)}$, where $\degree(u)\geq 1$ is the degree of $u$.
Note that $\Phi(X)\leq n\cdot \FitnessM_{\max}$.
We let $\Delta_t=\Phi(\RandomConfig_{t+1})-\Phi(\RandomConfig_{t})$ be the potential difference in step $t$.
In addition, let $\RandomConfig_t=\Config$, and $R=\{ (u,v)\in E \colon u\in \Config \text{ and } v\not \in \Config \}$ be the set of  edges in $\Config$ with one endpoint being mutant and the other being resident.
Moreover, denote $F=\sum_{u\in V}\Fitness_{\Config}(u)$ as the total population fitness in $\Config$. 
Given a pair $(u,v)\in R$, let $p_{u\to v}$ be the probability that $u$ reproduces and replaces $v$.
First we show that $\Expectation(\Delta_t)\geq 0$, i.e., in expectation, the potential function increases in each step:
\begin{align*}
\Expectation(\Delta_t)&=\sum_{(u,v)\in R} \left(p_{u\to v}\cdot \frac{\FitnessM(v)}{\degree(v)} - p_{v\to u} \cdot \frac{\FitnessM(u)}{\degree(u)}\right)\\
&=\sum_{(u,v)\in R} \left(\frac{\FitnessM(u)}{F}\frac{1}{\degree(u)}\frac{\FitnessM(v)}{\degree(v)}-\frac{\FitnessR(v)}{F}\frac{1}{\degree(v)}\frac{\FitnessM(u)}{\degree(u)}\right)\\
&=\sum_{(u,v)\in R} \frac{\FitnessM(u)(\FitnessM(v)-\FitnessR(v))}{{d(u)d(v)}F}\geq 0
\end{align*}
as $\FitnessM(v)\geq \FitnessR(v)\geq\MinResidentFitness$ since $\FitnessG$ is mutant-biased.
Second, we give a lower bound on the variance of $\Delta_t$ when $\emptyset\subset \Config\subset V$, and thus there exists an edge $(u,v)\in R$.
First, we have 
\begin{linenomath*}
\begin{align*}
p_{v\to u} = \frac{\FitnessR(v)}{F} \frac{1}{\degree(v)}\geq \frac{\MinResidentFitness}{n\cdot\MaxMutantFitness}\frac{1}{n} =\frac{\MinResidentFitness}{n^2\cdot\MaxMutantFitness}
\end{align*}
\end{linenomath*}
while the potential function changes by $\Delta_t\leq-\frac{\FitnessM(u)}{d(u)}$.
Therefore, $\Prob\left[\Delta_t\leq \sfrac{-\FitnessM(u)}{d(u)}\right]\geq\frac{\MinResidentFitness}{n^2\cdot\MaxMutantFitness}$, and
\begin{linenomath*}
\begin{align*}
     \Variance(\Delta_t)&\geq \Prob\left[\Delta_t\leq -\frac{\FitnessM(u)}{d(u)}\right]\cdot\left(-\frac{\FitnessM(u)}{d(u)}-\Expectation(\Delta_t)\right)^2\\
    &\geq\frac{\MinResidentFitness}{n^2\cdot\MaxMutantFitness}\left(-\frac{\MinResidentFitness}{n}\right)^2=\frac{\MinResidentFitness^3}{n^4\cdot\MaxMutantFitness}.
\end{align*}
\end{linenomath*}
The potential $\Phi$ gives rise to a submartingale with upper bound $B=n\cdot\MaxMutantFitness$.
The re-scaled function $\Phi(\Phi-2B)+B^2$ satisfies the conditions
of the upper additive drift theorem~\cite{Kotz2019} with initial value at most $B^2$ and step-wise drift at least $\Variance(\Delta_t)$. 
We thus arrive at
\begin{align*}
T(G,S) \leq \frac{B^2}{\Variance(\Delta_t)} = \frac{n^2\cdot \MaxMutantFitness^2}{\frac{\MinResidentFitness^3}{n^4\cdot\MaxMutantFitness}}
\leq \frac{n^6\cdot \MaxMutantFitness^3}{\MinResidentFitness^3}.
\end{align*}
\end{proof}

\cref{lem:expected_time} yields an FPRAS for the fixation probability when mutant and resident fitnesses are polynomially (in $n$) related.

\begin{corollary}\label{cor:fpras}
Given a mutant-biased undirected fitness graph $\FitnessG$ with ${\MaxMutantFitness}/{\MinResidentFitness}=n^{O(1)}$ and a seed set $\SeedSet\subseteq V$,
the fixation probability $\FP_{\FitnessG}(\SeedSet)$ admits an FPRAS.
\end{corollary}

\section{Hardness of Optimization}\label{sec:hardness}

Here we turn our attention to the seed selection problem, and prove two hardness results.
First, we show that on arbitrary graphs, for any $0<\varepsilon<\sfrac{1}{2}$, it is $\NP$-hard to distinguish between graphs that achieve maximum fixation probability at most $\varepsilon$ and at least $1-\varepsilon$.
This is in sharp contrast to standard cascade models of influence spread, for which the optimal spread can be efficiently approximated~\cite{kempe2003maximizing}.
Then we focus on mutant-biased graphs, and show that achieving the maximum fixation probability remains $\NP$-hard even in this restricted setting.

Our reduction is from the $\NP$-hard problem Set Cover~\cite{karp1972reducibility}.
Given an instance $(\Universe, \Sets, \SizeConst)$, where~$\Universe$ is a universe, $\Sets$ a set of subsets of~$\Universe$, and~$\SizeConst$ a size constraint, the task is to decide whether there exist~$\SizeConst$ subsets in~$\Sets$ that cover~$\Universe$.
Wlog, $\Universe=\bigcup_{A\in \Sets}A$.
We construct a fitness graph $\,\FitnessG=(G,(\FitnessM,\FitnessR))$ where $G = (V, E, \Weight)$ is a bipartite graph with two parts $V = V_1 \cup V_2$ with $V_1 = \Sets$ and $V_2 = \Universe$, and define the edge relation as 
$
E = \{(u,v) \in V_1 \times V_2 \colon v \in u\} \cup (V_2 \times V_1)
$
i.e., there is an edge~$(u, v)$ for each element~$v$ of~$\Universe$ that appears in the set~$u$ of $\Sets$, as well as all possible edges from~$V_2$ to~$V_1$.
The weight function is uniform:~$\Weight(u,v)=1/\degree(u)$ for each $(u,v)\in E$.
The resident fitness is $\FitnessR(u)=1$ for all $u\in V$.
The mutant fitness is parametric on two values $x\geq 1$ and $y\leq 1$ to be fixed later, with $\FitnessM(u)=x$ if $u\in V_1$ and $\FitnessM(u)=y$ if $u\in V_2$.
See \cref{fig:set-cover} for an illustration.

Our construction guarantees upper and lower bounds on the fixation probability depending on whether the seed set forms a set cover of $(\Universe, \Sets)$, as stated in the following lemma.

\begin{restatable}{lemma}{lemhardnessmain}\label{lem:hardness_main}
The following assertions hold.
\begin{compactenum}
\item\label{item:hardness_main_upper} If \emph{$\SeedSet$ is not a set cover}, then $\FP_{\FitnessG}(\SeedSet) \leq 1-\left(\frac{\sfrac{1}{n}}{\sfrac{1}{n}+(n-1)y}\right)^{n}$.
\item\label{item:hardness_main_lower} If \emph{$\SeedSet$ is a set cover}, then
\begin{linenomath*}
\begin{align*}
\FP_{\FitnessG}(\SeedSet) \geq 
\left(\frac{\frac{y}{n^2}\left(\frac{\sfrac{x}{n}}{\sfrac{x}{n}+n}\right)^{n}}{1-\left(1-\frac{y}{n^2}\right)\left(\frac{\sfrac{x}{n}}{\sfrac{x}{n}+n}\right)^{n}}\right)^n. 
\end{align*}
\end{linenomath*}
\end{compactenum}
\end{restatable}
Before we prove \cref{lem:hardness_main}, 
we show 
how \cref{lem:hardness_main} leads to the two hardness results of this section.
\begin{restatable}{theorem}{thmhardnessgeneral}\label{thm:hardness_general}
For any $0<\varepsilon< \sfrac{1}{2}$, it is $\NP$-hard to distinguish between instances with~$\max_{\SeedSet}\FP_{\FitnessG}(\SeedSet)\leq\varepsilon$ and those with~$\max_{\SeedSet}\FP_{\FitnessG}(\SeedSet)> 1-\varepsilon$.
\end{restatable}
\begin{proof}[Proof sketch]
We solve the inequalities of \cref{lem:hardness_main}, and obtain that there exist $y=1/O(n^3)$ and $x=O(n^{10})$ satisfying them. As both values are polynomial in $n$, this completes a polynomial reduction from Set Cover to seed selection.
\end{proof}

\begin{figure}[!t]
\includegraphics[width=0.444\textwidth]{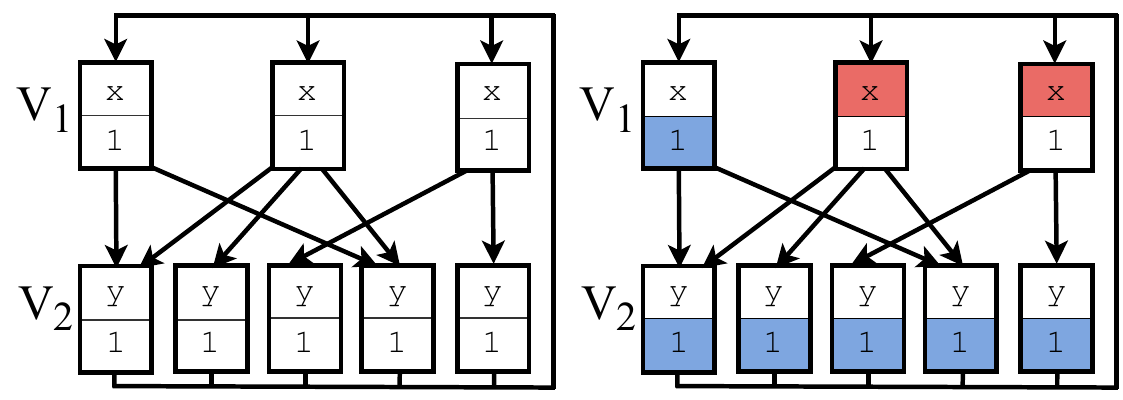}
\caption{\pk{\emph{(Left):}~Graph~$G$ for a Set Cover instance with~$\Universe=\{1,2,3,4,5\}$ and~$\Sets=\{ \{1,4\}, \{1,2,4\}, \{3,5\} \}$.
\emph{(Right):}~For~$\SizeConst=2$, the optimal seed set forms a Set Cover.}}\label{fig:set-cover}
\end{figure}

\begin{restatable}{theorem}{thmhardnessmutantbiased}\label{thm:hardness_mutant_biased}
For mutant-biased fitness graphs, it is $\NP$-hard to distinguish between instances with~$\max_{\SeedSet}\FP_{\FitnessG}(\SeedSet)\leq 1-1/(n^{2n})$ and those with~$\max_{\SeedSet}\FP_{\FitnessG}(\SeedSet)> 1-1/(n^{2n})$.
\end{restatable}
\begin{proof}[Proof sketch]
We set~$y=1$ and solve the second inequality of \cref{lem:hardness_main}.
We obtain that it is satisfied by some~$x = 2^{O(n\log n)}$. The fitness graph is mutant-biased as~$\FitnessR(u)=1$ and~$\FitnessM(u)\geq 1$ for all nodes~$u$. The description of~$x$ is polynomially long in~$n$, thus we have a polynomial reduction from Set Cover to seed selection on mutant-biased graphs.
\end{proof}

We remark that the class of graphs behind \cref{thm:hardness_mutant_biased} form a special case of the Positional Moran process~\cite{brendborg2022fixation}, by setting as active nodes $\ActiveNodes=V_1$ and fitness advantage $\delta=2^{O(n\log n)}$. 
Thus, the $\NP$-hardness of \cref{thm:hardness_mutant_biased} extends to the Positional Moran process.

We now turn our attention to the proof of \cref{lem:hardness_main}.
By a small abuse of terminology, we say that a configuration $\Config$ covers $V_2$ to denote that the sets in $\Config\cap V_1$ cover $V_2$.
\cref{item:hardness_main_upper} relies on the following intermediate lemma,
which intuitively states that, starting from a configuration $\Config_1$ that contains a resident node $v\in V_2$ not covered by $\Config$,
the process loses all mutants in $V_1$ with large enough probability.

\begin{restatable}{lemma}{lemupperresident}\label{lem:upper_resident}
From any configuration $\Config_1$ with $V_2\setminus (\Config_1 \cup \{ u\in \Config_1\colon (u,v)\in E \})\neq \emptyset$,
the process reaches a configuration $\Config_2$ with $\Config_2\cap V_1=\emptyset$ with probability $p\geq\left(\frac{\sfrac{1}{n}}{\sfrac{1}{n}+(n-1)y}\right)^{|V_1|}$.
\end{restatable}

We can now prove the upper bound of \cref{lem:hardness_main}.

\begin{proof}[Proof sketch of \cref{lem:hardness_main}, Item~1]
First, we show that the probability $q$ of reaching configuration $\Config_1$ such that $V_2\setminus (\Config_1 \cup \{ u\in \Config_1\colon (u,v)\in E \})\neq \emptyset$ is at least $\frac{\sfrac{1}{n}}{\sfrac{1}{n}+(n-1)y}$.
Then, by using \cref{lem:upper_resident} on $\Config_1$ we derive that the process reaches a configuration $\Config_2$ with $\Config_2\cap V_1=\emptyset$ with probability at least $\left(\frac{\sfrac{1}{n}}{\sfrac{1}{n}+(n-1)y}\right)^{|V_1|}=q^{|V_1|}$.
While at configuration $\Config_2$, the process changes configuration when either a resident in $V_1$ replaces a mutant in $V_2$, or vice versa.
Recall that the probability that the first event occurs before the second is at least $q$.
Repeating the process for all mutants in $V_2\setminus\{v\}$ (as $v$ is already a resident in $\Config_2$) we arrive in a configuration without mutants in $V_2$ with probability at least $q^{|V_2|-1}$.
At this point all mutants have gone extinct, thus:
\begin{align*}
    \FP_{\FitnessG}(\SeedSet)\leq 1- q^{1+|V_1|+(|V_2|-1)}
= 1-\left(\frac{\sfrac{1}{n}}{\sfrac{1}{n}+(n-1)y}\right)^{n}.
\end{align*}
\qedhere
\end{proof}

The following lemma 
states that, starting from a configuration $\Config$ that covers $V_2$, the process makes all nodes in $V_2$ mutants without losing any mutant in $V_1$, with certain probability.

\begin{restatable}{lemma}{lemlowermutant}\label{lem:lower_mutant}
From any configuration $\Config$ that covers $V_2$,
the process reaches a configuration $\Config^*$ with $V_2\cup (\Config\cap V_1)\subseteq \Config^*$ with probability $p^*\geq\left(\frac{\sfrac{x}{n}}{\sfrac{x}{n}+n}\right)^{|V_2|}$.
\end{restatable}

We can now prove the lower bound of \cref{lem:hardness_main}.

\begin{proof}[Proof sketch of \cref{lem:hardness_main}, Item~2]
We consider $4$ configurations; 
any configuration $\Config$ that covers~$V_2$;
$\Config^{-}$ with less mutants in $V_1$ than $\Config$;
$\Config^*$ with same mutants in $V_1$ with $\Config$ and all nodes in $V_2$ being mutants;
and
$\Config^{+}$ starting from $\Config^*$ includes at least one more mutant in $V_1$.
The Markov chain in \cref{fig:mc_np_hard} captures this process where states~$S_1$, $S_2$, $S_3$ and~$S_4$ denote that the process is in configurations~$\Config^{-}$, $\Config$, $\Config^*$ and~$\Config^+$, respectively. 
To prove the Lemma, we first bound the transition probabilities of Markov chain in \cref{fig:mc_np_hard}.
We prove that $p^+\geq \frac{\frac{y}{n^2}\left(\frac{\sfrac{x}{n}}{\sfrac{x}{n}+n}\right)^{n}}{1-\left(1-\frac{y}{n^2}\right)\left(\frac{\sfrac{x}{n}}{\sfrac{x}{n}+n}\right)^{n}}$, 
$p^*\geq\Big(\frac{\sfrac{x}{n}}{\sfrac{x}{n}+n}\Big)^{n}$ and $q\geq \frac{y}{n^2}$.
Note that $p^+$ is lower-bounded by the probability that a random walk starting in~$S_2$ (i.e., $\Config$) gets absorbed in~$S_4$ (i.e., $\Config^+$). 
Let~$x_i$ be the probability that a random walk starting in~$S_i$ gets absorbed in~$S_4$. 
We have
$x_2 = p^*\cdot x_3 + (1-p^*) \cdot x_1$
and $x_3 = q \cdot x_4 + (1-q) \cdot x_2$,
with boundary conditions~$x_1 = 0$ and~$x_4 = 1$, whence~$x_2 = \frac{q \cdot p^*}{1-(1-q) \cdot p^*}$. 
Since~$\Config \!\subset\! \Config^+$, the set $\Config^+$ also covers~$V_2$, thus the reasoning repeats for up to $n$ steps until fixation, resulting in 
$\FP_{\FitnessG}(\Config)\geq(p^+)^n$.
\end{proof}

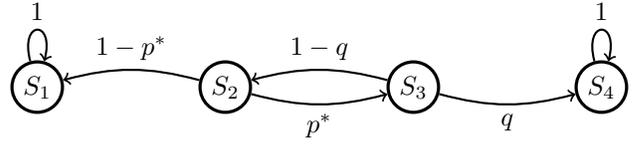
\begin{figure}
\centering
\begin{tikzpicture}
\def\xstep{2.5}
\def\bend{15}

\foreach \i in {1,...,4}{
\node[very thick, draw=black, circle, inner sep=2] (S\i) at (\i*\xstep, 0) {$S_{\i}$};
}

\draw[->, thick, bend right=\bend] (S2) to node[above]{$1-p^*$} (S1);
\draw[->, thick, bend right=\bend] (S2) to node[below]{$p^*$} (S3);

\draw[->, thick, bend right=\bend] (S3) to node[above]{$1-q$} (S2);
\draw[->, thick, bend right=\bend] (S3) to node[below]{$q$} (S4);

\draw[->, thick, loop above] (S1) to node[above]{$1$} (S1);
\draw[->, thick, loop above] (S4) to node[above]{$1$} (S4);

\end{tikzpicture}
\caption{\pk{The Markov chain for the process of \cref{lem:lower_mutant}.
}
}
\label{fig:mc_np_hard}
\end{figure}
\section{Monotonicity and Submodularity}\label{sec:submodularity}

\cref{thm:hardness_general} rules out polynomial-time algorithms for any non-trivial answer to seed selection. \cref{thm:hardness_mutant_biased} states that the problem remains $\NP$-hard for mutant-biased graphs, but it permits potential tractable approximations. Indeed, here we prove that mutant bias renders the fixation probability submodular, thus seed selection admits a constant-factor approximation. Our proofs are based on coupling arguments. Instead of applying these arguments directly to the Heterogeneous Moran process, we propose a variant, the \emph{Loopy process}, and show its equivalence to the Heterogeneous process in the sense of preserving the fixation probability.

\Paragraph{The Loopy process} 
In the Loopy process, we slightly modify the underlying fitness graph~$\FitnessG=(G,(\FitnessM, \FitnessR))$ in each step based on the current configuration~$\Config$. 
Without loss of generality, we let every node $u\in V$ have a self-loop $(u,u)\in E$, by assigning $\Weight(u,u) = 0$. Let $\Fitness_{\max}=\max_{u\in V}\{\FitnessR(u),\FitnessM(u)\}$ be the maximum fitness. When the original process is at some configuration $\Config$, different nodes reproduce at different rates. When the Loopy process is at configuration~$\Config$, we construct a fitness graph $\FitnessG_{\Config}=(G_{\Config}, (\ConstantOne, \ConstantOne))$, where $\ConstantOne$ is the constant function $u\mapsto 1$, and $G_{\Config}=(V,E,\Weight_{\Config})$ is a graph of the same structure as~$G$, but with a weight function modified by adjusting the self-loop probability of each node as follows.
\begin{align*}
\Weight_{\Config}(u,v)=
\begin{cases}
\frac{\Fitness_{\Config}(u)}{\Fitness_{\max}}\cdot\Weight(u,v),           & \text{if } u\neq v \\
1 - \frac{\Fitness_{\Config}(u)}{\Fitness_{\max}}(1-\Weight(u,v)), & \text{if } u = v
\end{cases}\numberthis\label{eq:loopyweights}
\end{align*}
where~$\Fitness_{\Config}(\cdot)$ denotes the fitness of nodes in the original process. \cref{fig:loopy} shows an instance of the Heterogeneous Moran process and its Loopy counterpart. All nodes in~$\FitnessG_{\Config}$ reproduce at equal rates as they have the same fitness regardless of type. The new weight function~$\Weight_{\Config}$ compensates for this reproduction rate uniformity:~nodes that formerly had lower fitness acquire stronger self-loops, hence the probability distribution~$\Prob[\RandomConfig_{t+1} |\RandomConfig_{t+1} \neq \RandomConfig_{t}]$, and thus the fixation probability, is identical in the two processes, as \cref{lem:loopyeq} states.

\begin{figure}[!b]
\centering
\includegraphics[width=0.46\textwidth]{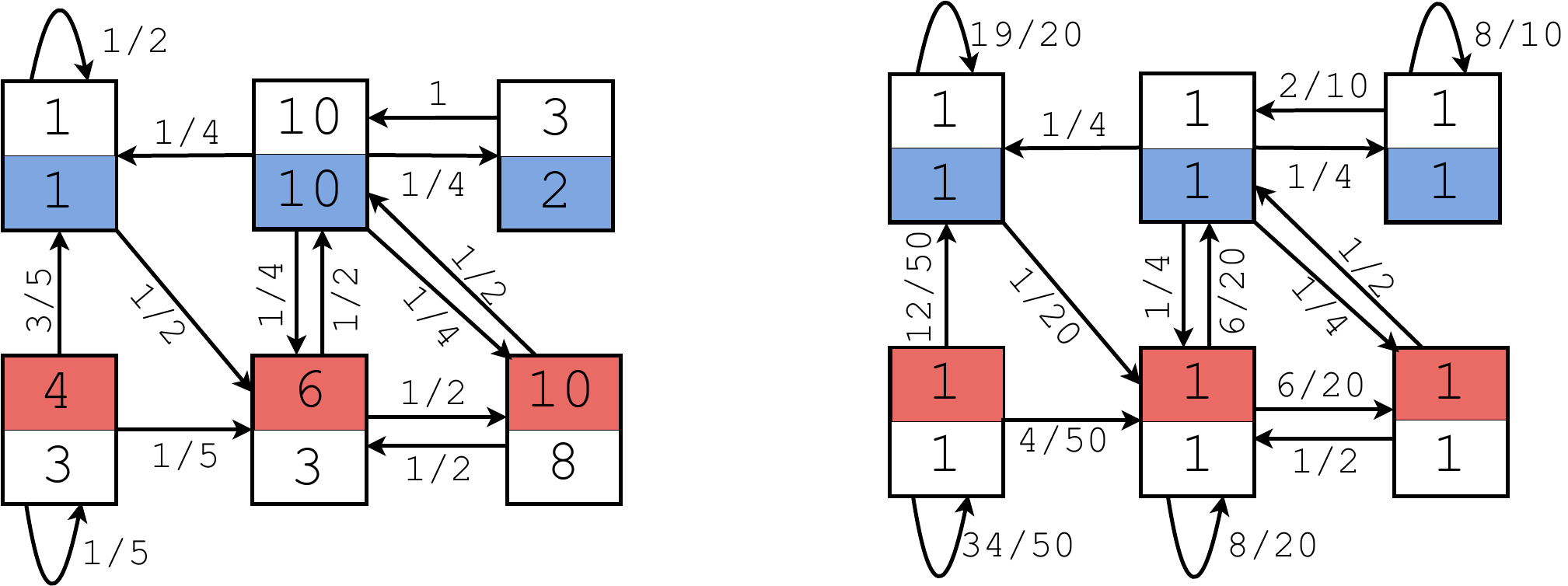}
\caption{
A fitness graph $\FitnessG$ (left) and its corresponding graph $G_{\Config}$ of the Loopy process (right).
All nodes in $G_{\Config}$ have the same fitness.}
\label{fig:loopy}
\end{figure}

\begin{restatable}{lemma}
{lemloopyeq}\label{lem:loopyeq}
For any seed set, the Heterogeneous and Loopy Moran processes share the same fixation probability.
\end{restatable}

\Paragraph{Relation to the Two-Graphs process}
The Loopy process is a special case of the recent Two-Graphs Moran process~\cite{melissourgos2022extension}. 
To obtain the Two-Graphs process,
we define two graphs $G_M$ and $G_R$ for mutants and residents, respectively.
For each edge $(u,v)$, its weight $w_M(u,v)$ in $G_M$ and $w_R(u,v)$ in $G_R$, is obtained from \cref{eq:loopyweights}, considering that $u\in \Config$ and $u\not \in \Config$, respectively.
In turn, \cref{lem:loopyeq} implies that the hardness of \cref{thm:hardness_general,thm:hardness_mutant_biased} also hold for seed selection in the Two-Graphs model.

\Paragraph{Monotonicity} 
The following monotonicity corollary follows from \cref{lem:loopyeq} and the monotonicity of the Two Graphs process~\cite[Corollary 6]{melissourgos2022extension}.

\begin{restatable}{corollary}
{lemmonotonicity}\label{lem:monotonicity}
For any fitness graph~$\FitnessG=(G,(\FitnessM, \FitnessR))$ and any two seed sets $\SeedSet \subseteq \SeedSet'$, we have~$\FP_{\FitnessG}(\SeedSet)\leq \FP_{\FitnessG}(\SeedSet')$.
\end{restatable}

\Paragraph{Submodularity} We now turn our attention to the submodularity of the fixation probability in the Heterogeneous Moran process.
Although the function is not submodular in general, we prove that it becomes submodular on mutant-biased fitness graphs.
In particular, we show that for any two seed sets $\SeedSet, \SeedSetT\subseteq V$, the following submodularity condition holds:
\begin{linenomath*}
\[
\FP_{\FitnessG}(\SeedSet) + \FP_{\FitnessG}(\SeedSetT) \geq \FP_{\FitnessG}(\SeedSet\cup \SeedSetT) + \FP_{\FitnessG}(\SeedSet\cap \SeedSetT)
\numberthis \label{eq:submodularity}
\]
\end{linenomath*}
Our proof is via a four-way coupling of the corresponding processes starting in one of the seed sets of \cref{eq:submodularity}.

\begin{restatable}{lemma}{lemsubmodularity}\label{lem:submodularity}
For any mutant-biased fitness graph~$\,\FitnessG=(G,(\FitnessM, \FitnessR))$, the fixation probability $\FP_{\FitnessG}(\SeedSet)$ is submodular.
\end{restatable}
\begin{proof}
Let~$\MP_1 = (\RandomConfig^1_{t})_{t \geq 0}$, $\MP_2 = (\RandomConfig^2_{t})_{t \geq 0}$, $\MP_3 = (\RandomConfig^3_{t})_{t \geq 0}$, and $\MP_4 = (\RandomConfig^4_{t})_{t \geq 0}$, be four Loopy processes with seed sets~$\SeedSet$, $\SeedSetT$, $\SeedSet \cup \SeedSetT$ and~$\SeedSet \cap \SeedSetT$, respectively. 
To prove submodularity, we employ two tricks for~$\MP_3$. 
First, along its configurations~$\RandomConfig^3_{t}$, we also keep track of the set of mutants~$\RandomConfigY_{t}$ (resp., $\RandomConfigZ_{t}$) that are copies of some initial node in~$\SeedSet$ (resp., $\SeedSetT$). 
Whenever a node~$v$ receives the mutant trait from a neighbor~$u$, we place~$v$ in~$\RandomConfigY_{t+1}$ (resp., $\RandomConfigZ_{t + 1}$) following the membership of~$u$ in~$\RandomConfigY_{t}$ (resp., $\RandomConfigZ_{t}$). 
Initially, $\RandomConfigY_0 = \SeedSet$ and~$\RandomConfigZ_0 = \SeedSetT$. 
Second, with probability~1, every run of~$\MP_3$ that results in fixation, eventually (i.e., if we let the process run on) leads to the fixation of~$\SeedSet$ or~$\SeedSetT$ 
(possibly both, assuming~$\SeedSet \cap \SeedSetT\neq \emptyset$); 
that is, every node is a copy of some node in~$\SeedSet$ or~$\SeedSetT$. 
We thus compute the fixation probability with seed~$\SeedSet \cup \SeedSetT$ by summing over runs in which~$\SeedSet$ or~$\SeedSetT$ fixates.

To prove submodularity, we consider this refined view of the process and establish a four-way coupling between~$\MP_1$, $\MP_2$, $\MP_3$ and~$\MP_4$ that guarantees the following invariants: (i)~$\RandomConfig^1_{t}\cup \RandomConfig^2_{t}\subseteq \RandomConfig^3_{t}$, (ii)~$\RandomConfig^4_{t} \subseteq\RandomConfig^1_{t}\cap \RandomConfig^2_{t}$, (iii)~$\RandomConfigY_{t} \subseteq \RandomConfig^1_{t}$, and (iv)~$\RandomConfigZ_{t}\subseteq \RandomConfig^2_{t}$. Now, consider any execution in which~$\MP_3$ fixates. Since~$\SeedSet$ or~$\SeedSetT$ eventually fixates in~$\MP_3$, due to invariants~(iii) and~(iv), at least one of~$\MP_1, \MP_2$ fixates as well. Moreover, if~$\MP_4$ also fixates, due to invariant~(ii),
both~$\MP_1$ and~$\MP_2$ fixate. Thus the invariants guarantee submodularity.

The invariants hold at~$t=0$. 
Now, consider some arbitrary time~$t$ with the four processes at configurations~$\RandomConfig^j_t = \Config^j$, for~$j\in \{1,2,3,4\}$, $\RandomConfigY_{t} = \ConfigY$, and~$\RandomConfigZ_{t} = \ConfigZ$. 
To obtain~$\RandomConfig^j_{t + 1}$, we sample the same node~$u$ for reproduction with probability~$\sfrac1n$ in all processes. 
From invariants~(i) and (ii), and since $\FitnessM(u)\geq\FitnessR(u)$, we derive that $\Weight_{\Config_3}(u,u) \leq \Weight_{\Config_j}(u,u) \leq \Weight_{\Config_4}(u,u)$ for~$j \in \{1, 2\}$, as residents have a larger self-loop weight.
In~$\MP_3$, we choose a neighbor~$v$ of~$u$ with probability~$\Weight_{\Config_3}(u, v)$ and propagate the trait of~$u$ to~$v$. 
In~$\MP_1$, $\MP_2$ and~$\MP_4$, if~$u = v$, we perform the same update; otherwise, if~$u$ has the same type as in~$\MP_3$, we also perform the same update. 
From the invariants, if~$u$ is resident in~$\MP_3$ then the same holds in all other processes, while if~$u$ is a mutant in~$\MP_3$ then the same holds in at least one of~$\MP_1$, $\MP_2$ (depending on whether~$u \in \ConfigY$ and~$u \in \ConfigZ$), and if that holds for~$\MP_1$ and~$\MP_2$, then it holds for~$\MP_4$. 
However, if $u$ is resident in~$\MP_j$ for some~$j\in \{1, 2\}$ but mutant in~$\MP_3$, i.e., $u \in \Config^3 \setminus \Config^j$, then, due to invariant~(ii), $u$ is also resident in~$\MP_4$, i.e., $u \in \Config^3 \setminus \Config^4$; then, in~$\MP_j$ and~$\MP_4$, $u$ propagates to itself with probability $\Weight_{\Config_j}(u,u)-\Weight_{\Config_3}(u,u)\geq 0$, 
and to~$v$ with the remaining probability~$1 - (\Weight_{\Config_j}(u,u) - \Weight_{\Config_3}(u,u))$. 
It follows that all three invariants are maintained.
\end{proof}

Following~\cite{nemhauser1978analysis}, monotonicity and submodularit lead to the following approximation guarantee.

\begin{theorem}\label{thm:apx}
Given a mutant-biased 
fitness graph $\FitnessG$ and budget $\SizeConst$, let~$\SeedSet^*$ be an optimal seed set and~$\SeedSet_{\text{gr}}$ the solution of the Greedy algorithm. 
We have $\FP_{\,\FitnessG}(\SeedSet_{\text{gr}}) \geq \left( 1 - \sfrac{1}{e} \right) \FP_{\,\FitnessG}(\SeedSet^{*})$.
\end{theorem}

The Greedy algorithm builds the seed set iteratively by choosing the node that yields the maximum fixation probability gain.
Finally, note that due to symmetry, on resident-biased fitness graphs  ($\FitnessM(u)\leq \FitnessR(u)$ for all $u$), 
$\FP_{\FitnessG}(\SeedSet)$ is \emph{supermodular}, thus Greedy offers no approximation guarantees.

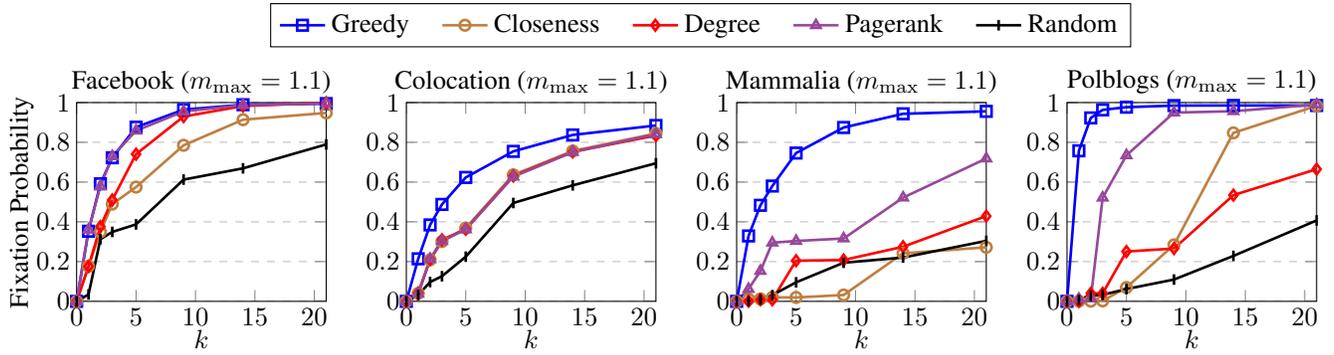
\begin{figure*}[ht]
\pgfplotsset{width=4.9cm,compat=1.3}

\begin{subfigure}[b]{.24\linewidth}
\begin{tikzpicture}
\pgfplotsset{every axis legend/.append style={
at={(2.5,1.5)},
anchor=north}}
\begin{axis}[
    legend columns=7,
    legend cell align={left},
    legend style={/tikz/every even column/.append style={column sep=0.35cm}},
    title={Facebook ($\MaxMutantFitness=1.1$)},
    title style={yshift=-6 pt,},
    xlabel = $k$,
    ylabel = {Fixation Probability},
    ymajorgrids=true,
    grid style=dashed,
    ylabel shift = -6 pt,
    xlabel shift = -5 pt,
    yticklabel shift = -1 pt,
    xticklabel shift = -1 pt,
    xmin=0, xmax=21,
    ymin=0, ymax=1
]
\addplot[
    color=blue,
    mark=square,
    mark options=solid,
    mark size=2pt,
    line width=1pt
    ] table[x = {k}, y = {FP}]{data/facebook/greedy_results_k.txt}; ]

\addplot[
    color=brown,
    mark = o,
    mark options=solid,
    mark size=2pt,
    line width=1pt
    ] table[x ={k}, y = {FP}]{data/facebook/MinClosenessCentrality_results_k.txt}; ]

\addplot[
    color=red,
    mark options=solid,
    mark = diamond,
    mark size=2pt,
    line width=1pt
    ] table[x={k}, y = {FP}]{data/facebook/MinDegreeCentrality_results_k.txt}; ]

\addplot[
    color=Purple,
    mark options=solid,
    mark = triangle,
    mark size=2pt,
    line width=1pt
    ] table[x ={k}, y = {FP}]{data/facebook/MinPageRank_results_k.txt}; ]

\addplot[
    color=black,
    mark options=solid,
    mark = |,
    mark size=2pt,
    line width=1pt
    ] table[x ={k}, y = {FP}]{data/facebook/Random_results_k.txt}; ]

\legend{Greedy, Closeness, Degree, Pagerank, Random}
\end{axis}
\end{tikzpicture}

\end{subfigure}
\hfill
\hspace{2.5mm}
\begin{subfigure}[b]{.24\linewidth}
\begin{tikzpicture}
\pgfplotsset{every axis legend/.append style={
at={(2.15,1.41)},
anchor=north}}
\begin{axis}[
    legend columns=7,
    legend cell align={left},
    legend style={/tikz/every even column/.append style={column sep=0.35cm}},
    title={Colocation ($\MaxMutantFitness=1.1$)},
    title style={yshift=-6 pt,},
    xlabel = $k$,
    ymajorgrids=true,
    grid style=dashed,
    ylabel shift = -6 pt,
    xlabel shift = -5 pt,
    yticklabel shift = -1 pt,
    xticklabel shift = -1 pt,
    xmin=0, xmax=21,
    ymin=0, ymax=1
]
\addplot[
    color=blue,
    mark=square,
    mark options=solid,
    mark size=2pt,
    line width=1pt
    ] table[x = {k}, y = {FP}]{data/colocation/greedy_results_k.txt}; ]

\addplot[
    color=brown,
    mark = o,
    mark options=solid,
    mark size=2pt,
    line width=1pt
    ] table[x ={k}, y = {FP}]{data/colocation/MinClosenessCentrality_results_k.txt}; ]

\addplot[
    color=red,
    mark options=solid,
    mark = diamond,
    mark size=2pt,
    line width=1pt
    ] table[x={k}, y = {FP}]{data/colocation/MinDegreeCentrality_results_k.txt}; ]

\addplot[
    color=Purple,
    mark options=solid,
    mark = triangle,
    mark size=2pt,
    line width=1pt
    ] table[x ={k}, y = {FP}]{data/colocation/MinPageRank_results_k.txt}; ]

\addplot[
    color=black,
    mark options=solid,
    mark = |,
    mark size=2pt,
    line width=1pt
    ] table[x ={k}, y = {FP}]{data/colocation/Random_results_k.txt}; ]

\end{axis}
\end{tikzpicture}
\end{subfigure}
\hfill
\begin{subfigure}[b]{.24\linewidth}
\begin{tikzpicture}
\pgfplotsset{every axis legend/.append style={
at={(2.15,1.41)},
anchor=north}}
\begin{axis}[
    legend columns=7,
    legend cell align={left},
    legend style={/tikz/every even column/.append style={column sep=0.35cm}},
    title={Mammalia ($\MaxMutantFitness=1.1$)},
    title style={yshift=-6 pt,},
    xlabel = $k$,
    ymajorgrids=true,
    grid style=dashed,
    ylabel shift = -6 pt,
    xlabel shift = -5 pt,
    yticklabel shift = -1 pt,
    xticklabel shift = -1 pt,
    xmin=0, xmax=21,
    ymin=0, ymax=1
]
\addplot[
    color=blue,
    mark=square,
    mark options=solid,
    mark size=2pt,
    line width=1pt
    ] table[x = {k}, y = {FP}]{data/mammalia/greedy_results_k.txt}; ]

\addplot[
    color=brown,
    mark = o,
    mark options=solid,
    mark size=2pt,
    line width=1pt
    ] table[x ={k}, y = {FP}]{data/mammalia/MinClosenessCentrality_results_k.txt}; ]

\addplot[
    color=red,
    mark options=solid,
    mark = diamond,
    mark size=2pt,
    line width=1pt
    ] table[x={k}, y = {FP}]{data/mammalia/MinDegreeCentrality_results_k.txt}; ]

\addplot[
    color=Purple,
    mark options=solid,
    mark = triangle,
    mark size=2pt,
    line width=1pt
    ] table[x ={k}, y = {FP}]{data/mammalia/MinPageRank_results_k.txt}; ]

\addplot[
    color=black,
    mark options=solid,
    mark = |,
    mark size=2pt,
    line width=1pt
    ] table[x ={k}, y = {FP}]{data/mammalia/Random_results_k.txt}; ]
\end{axis}
\end{tikzpicture}

\end{subfigure}
\hfill
\begin{subfigure}[b]{.24\linewidth}
\begin{tikzpicture}
\pgfplotsset{every axis legend/.append style={
at={(2.15,1.41)},
anchor=north}}
\begin{axis}[
    legend columns=7,
    legend cell align={left},
    legend style={/tikz/every even column/.append style={column sep=0.35cm}},
    title={Polblogs ($\MaxMutantFitness=1.1$)},
    title style={yshift=-6 pt,},
    xlabel = $k$,
    ymajorgrids=true,
    grid style=dashed,
    ylabel shift = -6 pt,
    xlabel shift = -5 pt,
    yticklabel shift = -1 pt,
    xticklabel shift = -1 pt,
    xmin=0, xmax=21,
    ymin=0, ymax=1
]
\addplot[
    color=blue,
    mark=square,
    mark options=solid,
    mark size=2pt,
    line width=1pt
    ] table[x = {k}, y = {FP}]{data/polblogs/greedy_results_k.txt}; ]

\addplot[
    color=brown,
    mark = o,
    mark options=solid,
    mark size=2pt,
    line width=1pt
    ] table[x ={k}, y = {FP}]{data/polblogs/MinClosenessCentrality_results_k.txt}; ]

\addplot[
    color=red,
    mark options=solid,
    mark = diamond,
    mark size=2pt,
    line width=1pt
    ] table[x={k}, y = {FP}]{data/polblogs/MinDegreeCentrality_results_k.txt}; ]

\addplot[
    color=Purple,
    mark options=solid,
    mark = triangle,
    mark size=2pt,
    line width=1pt
    ] table[x ={k}, y = {FP}]{data/polblogs/MinPageRank_results_k.txt}; ]

\addplot[
    color=black,
    mark options=solid,
    mark = |,
    mark size=2pt,
    line width=1pt
    ] table[x ={k}, y = {FP}]{data/polblogs/Random_results_k.txt}; ]

\end{axis}
\end{tikzpicture}

\end{subfigure}

\caption{\label{fig:fp_k}Fixation probability vs. $\SizeConst$.
}
\end{figure*}
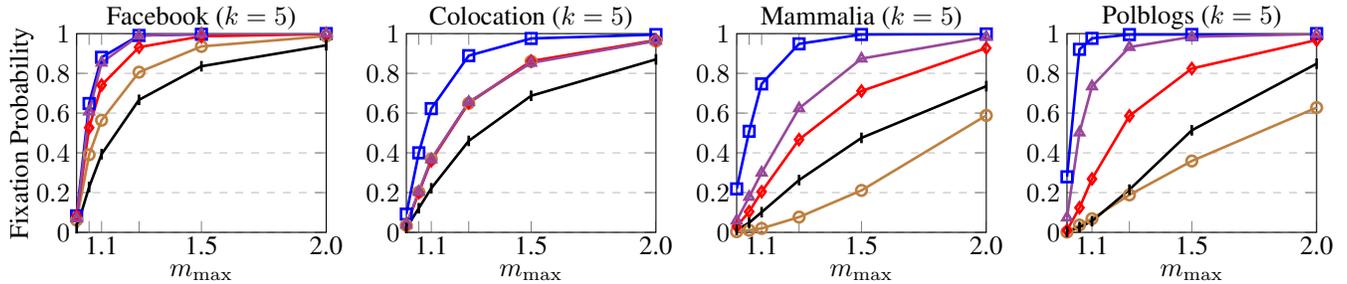
\begin{figure*}[ht]
\pgfplotsset{width=4.9cm,compat=1.3}

\begin{subfigure}[b]{.24\linewidth}
\begin{tikzpicture}
\pgfplotsset{every axis legend/.append style={
at={(2.15,1.41)},
anchor=north}}
\begin{axis}[
    legend columns=7,
    legend cell align={left},
    legend style={/tikz/every even column/.append style={column sep=0.35cm}},
    title={Facebook ($k=5$)},
    title style={yshift=-8 pt,},
    xlabel = $\MaxMutantFitness$,
    ylabel = {Fixation Probability},
    ymajorgrids=true,
    grid style=dashed,
    ylabel shift = -6 pt,
    xlabel shift = -3 pt,
    yticklabel shift = -1 pt,
    xticklabel shift = -2 pt,
    xmin=1, xmax=2,
    ymin=0, ymax=1,
    xtick = {1, 1.05, 1.1, 1.25, 1.5, 2.0},
    xticklabels={, , 1.1,, 1.5, 2.0}
]
\addplot[select coords between index={0}{5}][
    color=blue,
    mark=square,
    mark options=solid,
    mark size=2pt,
    line width=1pt
    ] table[x expr={\thisrowno{1}}, y = {FP}]{data/facebook/greedy_results_r.txt}; ]

\addplot[select coords between index={0}{5}][
    color=brown,
    mark = o,
    mark options=solid,
    mark size=2pt,
    line width=1pt
    ] table[x expr={\thisrowno{1}}, y = {FP}]{data/facebook/MinClosenessCentrality_results_r.txt}; ]

\addplot[select coords between index={0}{5}][
    color=red,
    mark = diamond,
    mark options=solid,
    mark size=2pt,
    line width=1pt
    ] table[x expr={\thisrowno{1}}, y = {FP}]{data/facebook/MinDegreeCentrality_results_r.txt}; ]

\addplot[select coords between index={0}{5}][
    color=Purple,
    mark = triangle,
    mark options=solid,
    mark size=2pt,
    line width=1pt
    ] table[x expr={\thisrowno{1}}, y = {FP}]{data/facebook/MinPageRank_results_r.txt}; ]

\addplot[select coords between index={0}{5}][
    color=black,
    mark = |,
    mark options=solid,
    mark size=2pt,
    line width=1pt
    ] table[x expr={\thisrowno{1}}, y = {FP}]{data/facebook/Random_results_r.txt}; ]

\end{axis}
\end{tikzpicture}

\end{subfigure}
\hfill
\hspace{2.5mm}
\begin{subfigure}[b]{.24\linewidth}
\begin{tikzpicture}
\pgfplotsset{every axis legend/.append style={
at={(2.15,1.41)},
anchor=north}}
\begin{axis}[
    legend columns=7,
    legend cell align={left},
    legend style={/tikz/every even column/.append style={column sep=0.35cm}},
    title={Colocation ($k=5$)},
    title style={yshift=-8 pt,},
    xlabel = $\MaxMutantFitness$,
    ymajorgrids=true,
    grid style=dashed,
    ylabel shift = -6 pt,
    xlabel shift = -3 pt,
    yticklabel shift = -1 pt,
    xticklabel shift = -2 pt,
    xmin=1, xmax=2,
    ymin=0, ymax=1,
    xtick = {1, 1.05, 1.1, 1.25, 1.5, 2.0},
    xticklabels={, , 1.1, , 1.5, 2.0}
]
\addplot[select coords between index={0}{5}][
    color=blue,
    mark=square,
    mark options=solid,
    mark size=2pt,
    line width=1pt
    ] table[x expr={\thisrowno{1}}, y = {FP}]{data/colocation/greedy_results_r.txt}; ]

\addplot[select coords between index={0}{5}][
    color=brown,
    mark = o,
    mark options=solid,
    mark size=2pt,
    line width=1pt
    ] table[x expr={\thisrowno{1}}, y = {FP}]{data/colocation/MinClosenessCentrality_results_r.txt}; ]

\addplot[select coords between index={0}{5}][
    color=red,
    mark = diamond,
    mark options=solid,
    mark size=2pt,
    line width=1pt
    ] table[x expr={\thisrowno{1}}, y = {FP}]{data/colocation/MinDegreeCentrality_results_r.txt}; ]

\addplot[select coords between index={0}{5}][
    color=Purple,
    mark = triangle,
    mark options=solid,
    mark size=2pt,
    line width=1pt
    ] table[x expr={\thisrowno{1}}, y = {FP}]{data/colocation/MinPageRank_results_r.txt}; ]

\addplot[select coords between index={0}{5}][
    color=black,
    mark = |,
    mark options=solid,
    mark size=2pt,
    line width=1pt
    ] table[x expr={\thisrowno{1}}, y = {FP}]{data/colocation/Random_results_r.txt}; ]

\end{axis}
\end{tikzpicture}

\end{subfigure}
\hfill
\begin{subfigure}[b]{.24\linewidth}
\begin{tikzpicture}
\pgfplotsset{every axis legend/.append style={
at={(2.15,1.41)},
anchor=north}}
\begin{axis}[
    legend columns=7,
    legend cell align={left},
    legend style={/tikz/every even column/.append style={column sep=0.35cm}},
    title={Mammalia ($k=5$)},
    title style={yshift=-8 pt,},
    xlabel = $\MaxMutantFitness$,
    ymajorgrids=true,
    grid style=dashed,
    ylabel shift = -6 pt,
    xlabel shift = -3 pt,
    yticklabel shift = -1 pt,
    xticklabel shift = -2 pt,
    xmin=1, xmax=2,
    ymin=0, ymax=1,
    xtick = {1, 1.05, 1.1, 1.25, 1.5, 2.0},
    xticklabels={, , 1.1, , 1.5, 2.0}
]
\addplot[select coords between index={0}{5}][
    color=blue,
    mark=square,
    mark options=solid,
    mark size=2pt,
    line width=1pt
    ] table[x expr={\thisrowno{1}}, y = {FP}]{data/mammalia/greedy_results_r.txt}; ]

\addplot[select coords between index={0}{5}][
    color=brown,
    mark = o,
    mark options=solid,
    mark size=2pt,
    line width=1pt
    ] table[x expr={\thisrowno{1}}, y = {FP}]{data/mammalia/MinClosenessCentrality_results_r.txt}; ]

\addplot[select coords between index={0}{5}][
    color=red,
    mark = diamond,
    mark options=solid,
    mark size=2pt,
    line width=1pt
    ] table[x expr={\thisrowno{1}}, y = {FP}]{data/mammalia/MinDegreeCentrality_results_r.txt}; ]

\addplot[select coords between index={0}{5}][
    color=Purple,
    mark = triangle,
    mark options=solid,
    mark size=2pt,
    line width=1pt
    ] table[x expr={\thisrowno{1}}, y = {FP}]{data/mammalia/MinPageRank_results_r.txt}; ]

\addplot[select coords between index={0}{5}][
    color=black,
    mark = |,
    mark options=solid,
    mark size=2pt,
    line width=1pt
    ] table[x expr={\thisrowno{1}}, y = {FP}]{data/mammalia/Random_results_r.txt}; ]

\end{axis}
\end{tikzpicture}

\end{subfigure}
\hfill
\begin{subfigure}[b]{.24\linewidth}
\begin{tikzpicture}
\pgfplotsset{every axis legend/.append style={
at={(2.15,1.41)},
anchor=north}}
\begin{axis}[
    legend columns=7,
    legend cell align={left},
    legend style={/tikz/every even column/.append style={column sep=0.35cm}},
    title={Polblogs ($k=5$)},
    title style={yshift=-8 pt,},
    xlabel = $\MaxMutantFitness$,
    ymajorgrids=true,
    grid style=dashed,
    ylabel shift = -6 pt,
    xlabel shift = -3 pt,
    yticklabel shift = -1 pt,
    xticklabel shift = -2 pt,
    xmin=1, xmax=2,
    ymin=0, ymax=1,
    xtick = {1, 1.05, 1.1, 1.25, 1.5, 2.0},
    xticklabels={, , 1.1, , 1.5, 2.0}
]
\addplot[select coords between index={0}{5}][
    color=blue,
    mark=square,
    mark options=solid,
    mark size=2pt,
    line width=1pt
    ] table[x expr={\thisrowno{1}}, y = {FP}]{data/polblogs/greedy_results_r.txt}; ]

\addplot[select coords between index={0}{5}][
    color=brown,
    mark = o,
    mark options=solid,
    mark size=2pt,
    line width=1pt
    ] table[x expr={\thisrowno{1}}, y = {FP}]{data/polblogs/MinClosenessCentrality_results_r.txt}; ]

\addplot[select coords between index={0}{5}][
    color=red,
    mark = diamond,
    mark options=solid,
    mark size=2pt,
    line width=1pt
    ] table[x expr={\thisrowno{1}}, y = {FP}]{data/polblogs/MinDegreeCentrality_results_r.txt}; ]

\addplot[select coords between index={0}{5}][
    color=Purple,
    mark = triangle,
    mark options=solid,
    mark size=2pt,
    line width=1pt
    ] table[x expr={\thisrowno{1}}, y = {FP}]{data/polblogs/MinPageRank_results_r.txt}; ]

\addplot[select coords between index={0}{5}][
    color=black,
    mark = |,
    mark options=solid,
    mark size=2pt,
    line width=1pt
    ] table[x expr={\thisrowno{1}}, y = {FP}]{data/polblogs/Random_results_r.txt}; ]

\end{axis}
\end{tikzpicture}

\end{subfigure}
\caption{\label{fig:fp_delta_k}Fixation probability vs. $\MaxMutantFitness$.
}
\end{figure*}

\section{Experimental Analysis}\label{sec:experiments}

Here, we present our experimental evaluation of the Greedy algorithm and other network heuristics, 
varying the seed size $k$ and the maximum mutant fitness $m_{\max}$.

\begin{table}[!b]
{
\centering
 \begin{tabular}{||c c c  c c||} 
 \hline
 Name & $|V|$  & $|E|$ &  Directed & Edge-Weighted  \\ [-0.1ex] 
 \hline
  Facebook & 324 & 5028  & \xmark & \xmark  \\
  Colocation & 242 & 53188 & \xmark & \cmark  \\ 
   Mammalia & 327 & 1045 &  \cmark & \cmark  \\
 Polblogs & 793 & 15839 &  \cmark & \xmark  \\ 
 \hline
 \end{tabular}
\caption{Dataset characteristics.}\label{tab:datasets}
}
\end{table}

\Paragraph{Datasets} We use four real-world networks from  Netzschleuder, SNAP and Network Repository (\cref{tab:datasets}).

\begin{compactenum}
\item \emph{Facebook:} A Facebook ego network in which nodes represent profiles and edges indicate friendship.
\item \emph{Colocation:} A proximity network of students and teachers of a French school.
Edge weights count
the frequency of contact between individuals
during a two-day period.
\item \emph{Mammalia:} An animal-contact network based on movements of voles (\emph{Microtus agrestis}).
Each edge weight counts the common traps the two voles were caught in.
\item \emph{Polblogs:}
A network of hyperlinks among a large set of U.S. political weblogs from before the 2004 election. 
\end{compactenum}
Our experiments are not meant to be exhaustive, but rather indicative of the performance of the greedy algorithm and common network-optimization heuristics on a few diverse networks.
We set the resident fitness to $1$,
while the mutant fitness of each node $u$ is determined by sampling a uniform distribution $\FitnessM(u)\sim \mathcal{U}(1,\MaxMutantFitness)$.
This results in mutant-biased graphs, for which \cref{thm:apx} guarantees that the fixation probability admits a Monte Carlo approximation.

\Paragraph{Greedy and Baselines}
We evaluate the performance of the standard Greedy algorithm behind \cref{thm:apx}~\cite{nemhauser1978analysis} 
against four common baseline algorithms from related literature on seed selection under diffusion processes \cite{brendborg2022fixation,zhao2021structural,liu2017influence}.
\begin{compactenum}
\item \emph{Random:} select uniformly at random.
\item \emph{Degree:} select by smallest degree.
\item \emph{Closeness:} select by smallest closeness centrality.
\item \emph{PageRank:} select by smallest PageRank score.
\end{compactenum}

The Random selection strategy is a standard baseline to measure the intricacy of the problem.
Degree is the only existing algorithm for seed selection in the Moran model, and is optimal for undirected and unweighted networks under the neutral setting
(but underperforms when $\MaxMutantFitness>1$).
On the other hand, Closeness and PageRank take into account the structure of the graph and its connectivity.
For these two centrality heuristics we also tried selecting the top-$k$-nodes by largest value, which resulted in worse performance.
All Monte Carlo simulations were run over $5000$ iterations.

{\Paragraph{Performance vs. $\bm{\SizeConst}$} \cref{fig:fp_k} shows performance as the size constraint~$k$ increases for a fixed mutant fitness distribution. 
In agreement with \cref{lem:monotonicity,lem:submodularity}, the performance of all algorithms rises as~$k$ grows, while Greedy has diminishing returns. 
Notably, Greedy outperforms all heuristics especially for small size constraints, while  Pagerank forms high quality solutions for the undirected and unweighted graph Facebook. 
On the other hand, seed selection becomes more challenging for directed (Mammalia, Polblogs) and edge-weighted graphs (Colocation, Mammalia), in which only Greedy uncovers high-quality seed sets.
}

\Paragraph{Performance vs. $\bm{\FitnessM}$} \cref{fig:fp_delta_k} shows performance as the mutant fitness interval $[1,\MaxMutantFitness]$ increases, for fixed size $\SizeConst$.
Random selection performs poorly, showing that the problem is not trivial,
while the other two heuristics have mixed performance.
On the other hand, Greedy achieves
a steady, high-quality performance in all datasets and problem parameters.
\section{Conclusion}\label{sec:conclusion}
We studied a natural optimization problem pertaining to network diffusion by the Heterogeneous Moran process, namely selecting a set of seed nodes that maximize the effect of the invasion.
To our knowledge, this is the first paper to study this standard optimization problem on Moran models.
We showed that the problem is strongly inapproximable in general, but becomes 
approximable 
on mutant-biased graphs, although 
the exact solution remains $\NP$-hard.
Several interesting questions remain open for future work, such as,
is seed selection hard in the Standard model; 
and are there tighter approximations for mutant-biased graphs?

\clearpage
\section*{Acknowledgments}
Work supported by grants from DFF (P.P. and P.K., 9041-00382B), Villum Fonden (A.P., VIL42117), and Charles University (J.T., UNCE 24/SCI/008 and PRIMUS 24/SCI/012).
\bibliographystyle{named_arxiv}
\bibliography{het_moran_arxiv}
\clearpage
\appendix
\appendix
\section{Appendix}\label{sec:app}

First we prove \cref{lem:upper_resident} and \cref{lem:lower_mutant}, that are used in the analysis of \cref{lem:hardness_main}.

\lemupperresident*
\begin{proof}
Let $q=\frac{\sfrac{1}{n}}{\sfrac{1}{n}+(n-1)y}$.
While at $\Config_1$, there exists at least one resident node $v\in V_2$ that is not covered by any mutant in $V_1$.
Let $\Config'_1$ be the first configuration that the process reaches in which the number of mutants in $V_1$ has changed from $\Config_1$
(i.e., $V_1$ either has one more, or one less mutant in $\Config'_1$ compared to $\Config_1$).
Let $F=\sum_{u\in V}\Fitness_{\Config_1}(u)$ be the total population fitness at $\Config_1$.
The probability that $v$ replaces a mutant in $V_1$ in a single step is $p_1\geq\frac{1}{F}\frac{1}{n}$.
On the other hand, the probability that any mutant in $V_2$ replaces a resident in $V_1$ in a single step is $p_2\leq\frac{|V_2| y}{F}\leq\frac{(n-1) y}{F}$.
Thus, the probability that $V_1$ has lost a mutant in $\Config'_1$ is at least $\frac{p_1}{p_1+p_2}\geq\frac{\sfrac{1}{n}}{\sfrac{1}{n}+(n-1)y}=q$.

Now, observe that $\Config'_1$ satisfies the conditions of $\Config_1$, i.e., $V_2\setminus (\Config'_1 \cup \{ u\in \Config'_1\colon (u,v)\in E \})\neq \emptyset$.
Thus we can repeat the above process until all nodes in $V_1$ have become residents, leading to the desired configuration $\Config_2$ with probability
\begin{linenomath*}
\[
p\geq q^{|V_1|}=\left(\frac{\sfrac{1}{n}}{\sfrac{1}{n}+(n-1)y}\right)^{|V_1|}\qedhere
\]
\end{linenomath*}
\end{proof}

\lemlowermutant*
\begin{proof}
Let $q=\frac{x/n}{x/n+n}$, and let $F=\sum_{u\in V}\Fitness_{\Config}(u)$ be the total population fitness in $\Config$.
As long as there are residents in $V_2$, the probability that a mutant node in $V_1$ replaces a resident in $V_2$ is at least $p_1\geq \frac{x}{F} \frac{1}{n}$.
On the other hand, the probability that a resident in $V_2$ (resp.\ $V_1$)
replaces a mutant in $V_1$ (resp. $V_2$) is at most $p_2\leq \frac{|V_2|}{F}$ (resp.\ $p_3\leq \frac{|V_1|}{F}$).
The probability that the first event occurs before the second and third one is thus at least $\frac{p_1}{p_1+p_2+p_3}\ge \frac{x/n}{x/n+|V_2|+|V_1|}=\frac{x/n}{x/n+n}=q$.
Observe that any other event that changes the configuration (i.e., mutants in $V_2$ replacing residents in $V_1$) 
results in a configuration that also covers $V_2$, 
thus we can repeat the above argument until all nodes in $V_2$ become mutants, which occurs with probability at least
\begin{linenomath*}
\[
p^*\geq q^{|V_2|}\ge\left(\frac{\sfrac{x}{n}}{\sfrac{x}{n}+n}\right)^{|V_2|}\qedhere
\]
\end{linenomath*}
\end{proof}

 We continue with proving the upper and lower bounds of  \cref{lem:hardness_main}.
\lemhardnessmain*
\begin{proof}
We prove the two assertions separately.
\begin{compactenum}
    \item{Item~1}:
    First, we argue that with probability at least $q=\frac{\sfrac{1}{n}}{\sfrac{1}{n}+(n-1)y}$, the process reaches a configuration $\Config_1$ such that  $V_2\setminus (\Config_1 \cup \{ u\in \Config_1\colon (u,v)\in E \})\neq \emptyset$.
    Indeed, since $\SeedSet$ does not form a set cover, there exists a node $v\in V_2$ that has no mutant incoming neighbor in $V_1$.
    If $v\not \in \SeedSet$, we are done.
    Otherwise, the probability that, in a single step, $v$ is replaced by any resident neighbor in $V_1$ is at least $p_1\geq\frac{1}{F} \frac{1}{n}$, where $F=\sum_{u\in V}\Fitness_{\SeedSet}(u)$ is the total population fitness at $\SeedSet$.
    On the other hand, any resident in $V_1$ is replaced by mutants in $V_2$ with probability $p_2\leq \frac{|V_2|y}{F}\leq \frac{(n-1)y}{F}$.
    Thus, the probability that the process reaches a desired configuration $\Config_1$ is at least $\frac{p_1}{p_1+p_2}\geq\frac{1/n}{1/n+(n-1)y}=q$.
    
    Second, \cref{lem:upper_resident} applies on $\Config_1$ to show that the process reaches a configuration $\Config_2$ with $\Config_2\cap V_1=\emptyset$ with probability at least $\left(\frac{\sfrac{1}{n}}{\sfrac{1}{n}+(n-1)y}\right)^{|V_1|}=q^{|V_1|}$.
    
    Third, while at configuration $\Config_2$, the process changes configuration when either a resident in $V_1$ replaces a mutant in $V_2$, or vice versa.
    We have already argued in the first step that the probability that the first event occurs before the second is at least $q$.
    Now, we repeat this process until all mutants in $V_2$ have become residents, which occurs with probability at least $q^{|V_2|-1}$, as $v\in V_2$ is already a resident in $\Config_2$.
    At this point the mutants have gone extinct, thus
    $
    \FP_{\FitnessG}(\SeedSet)\leq 1- q^{1+|V_1|+(|V_2|-1)}
    = 1-\left(\frac{\sfrac{1}{n}}{\sfrac{1}{n}+(n-1)y}\right)^{n}.
    $

    \item Item~2: We first prove the following statement. 
    Consider the process at any configuration~$\Config$ that covers~$V_2$, and let~$p^+$ be the probability that it reaches a subsequent configuration~$\Config^+$ with at least one more mutant in $V_1$. 
    We will show that $p^+\geq \frac{\frac{y}{n^2}\left(\frac{\sfrac{x}{n}}{\sfrac{x}{n}+n}\right)^{n}}{1-\left(1-\frac{y}{n^2}\right)\left(\frac{\sfrac{x}{n}}{\sfrac{x}{n}+n}\right)^{n}}$.
    Given any such configuration~$\Config$, let~$p^*$ be the probability that the process reaches a subsequent configuration~$\Config^*$ with $V_2 \cup (\Config\cap V_1)\subseteq \Config^*$. 
    By \cref{lem:lower_mutant}, we have $p^*\geq\Big(\frac{\sfrac{x}{n}}{\sfrac{x}{n}+n}\Big)^{|V_2|}\geq\Big(\frac{\sfrac{x}{n}}{\sfrac{x}{n}+n}\Big)^{n}$.
    While at~$\Config^*$, 
    the process only progresses 
    when a mutant in $V_2$ replaces a resident in $V_1$, or a resident in $V_1$ replaces a mutant in $V_2$. 
    As long as there are residents in $V_1$, the probability $q$ that a mutant from $V_2$ replaces a resident in $V_1$ before any such resident reproduces satisfies
    $q\geq\frac{\frac{y}{F}\frac{1}{n}}{\frac{y}{F}\frac{1}{n}+\frac{|V_1|}{F}}\geq\frac{\frac{y}{F}\frac{1}{n}}{\frac{y}{F}\frac{1}{n}+\frac{n-1}{F}}=\frac{y}{y+(n-1)n}\geq\frac{y}{n^2}$, where the last inequality holds as $y\leq 1\leq n$. 
    If this happens, we reach the desired configuration~$\Config^+$.
    Otherwise, 
    a resident in $V_1$ replaces a mutant in $V_2$, the resulting configuration still covers~$V_2$, and the argument repeats. 
    The Markov chain in \cref{fig:mc_np_hard} captures this process. 
    The states~$S_2$, $S_3$ and~$S_4$ denote that the process is in configurations~$\Config$, $\Config^*$ and~$\Config^+$, respectively. 
    Hence, $p^+$ is lower-bounded by the probability that a random walk starting in~$S_2$ (i.e., $\Config$) gets absorbed in~$S_4$ (i.e., $\Config^+$). 
    Let~$x_i$ be the probability that a random walk starting in~$S_i$ gets absorbed in~$S_4$. 
    We have
    $x_2 = p^*\cdot x_3 + (1-p^*) \cdot x_1$
    and $x_3 = q \cdot x_4 + (1-q) \cdot x_2$,
    with boundary conditions~$x_1 = 0$ and~$x_4 = 1$, whence~$x_2 = \frac{q \cdot p^*}{1-(1-q) \cdot p^*}$. 
    Since~$\Config \!\subset\! \Config^+$, set $\Config^+$ also covers~$V_2$, thus the reasoning applies for up to $n$ steps until fixation, resulting in 
    $\FP_{\FitnessG}(\Config)\geq(p^+)^n\geq
    \left(\frac{\frac{y}{n^2}\Big(\frac{\sfrac{x}{n}}{\sfrac{x}{n}+n}\Big)^{n}}{1-\Big(1-\frac{y}{n^2}\Big)\Big(\frac{\sfrac{x}{n}}{\sfrac{x}{n}+n}\Big)^{n}}\right)^n.$\qedhere
\end{compactenum}

\end{proof}

We continue with the proof of \cref{thm:hardness_general} and \cref{thm:hardness_mutant_biased}.
To make our analysis easier, we first prove two simple lemmas.

\begin{restatable}{lemma}{loglowerbound}\label{lem:log_lower_bound}
For every $\zeta>0$, we have $\ln\left(1+\frac{1}{\zeta}\right)\geq \frac{1}{\zeta+1}$.
\end{restatable}
\begin{proof}
\begin{linenomath*}
\begin{align*}
&\left(1+\frac{1}{\zeta}\right)^{\zeta+1}\geq e\\
\Leftrightarrow&\ln\left(1+\frac{1}{\zeta}\right)^{\zeta+1}\geq \ln e\\
\Leftrightarrow&\ln\left(1+\frac{1}{\zeta}\right)\geq \frac{1}{\zeta+1}\qedhere
\end{align*}
\end{linenomath*}
\end{proof}

\begin{restatable}{lemma}{lemprobbound}\label{lem:prob_bound}
Let $p\in(0,1)$, and $\beta\leq \frac{\ln(\sfrac{1}{p})}{n}$.
We have
\begin{linenomath*}
\begin{align*}
\left( \frac{1}{1+\beta}\right)^n\geq p
\end{align*}
\end{linenomath*}
\end{restatable}
\begin{proof}
Let $\gamma=1+\sfrac{1}{\beta}$, thus $\beta=\sfrac{1}{(\gamma-1)}$.
Then 
\begin{linenomath*}
\begin{align*}
\gamma = 1 + \frac{1}{\beta} \geq 1+ \frac{n}{\ln\left(\frac{1}{p}\right)}
\end{align*}
\end{linenomath*}
Moreover
\begin{linenomath*}
\begin{align*}
\frac{1}{1+\beta} = \frac{1}{1+\frac{1}{\gamma-1}}=\frac{1}{\frac{\gamma}{\gamma-1}}=\frac{\gamma-1}{\gamma}=1-\frac{1}{\gamma}
\end{align*}
\end{linenomath*}
Let $k=\ln(\sfrac{1}{p})^{-1}$, thus $p=e^{-\sfrac{1}{k}}$ and $\gamma\geq nk+1$.
We have
\begin{linenomath*}
\begin{align*}
&\left(1-\frac{1}{nk + 1}\right)^{nk} \geq e^{-1}\\
\Rightarrow& \left(1-\frac{1}{nk + 1}\right)^n \geq e^{-\frac{1}{k}}\\
\Rightarrow& \left( 1-\frac{1}{\gamma}\right)^n\geq p \\
\Rightarrow & \left( \frac{1}{1+\beta} \right)^n \geq p
\end{align*}
\end{linenomath*}
as desired.
\end{proof}

\thmhardnessgeneral*
\begin{proof}
The proof  is by algebraic manipulation on the inequalities of \cref{lem:hardness_main}.
In particular, we argue that there exist $x$ and $y$ that have polynomially-long description (in $n$) for which the inequalities stated in the theorem hold.
In turn, this completes a polynomial reduction from Set Cover to distinguishing between $\max_{\SeedSet}\FP_{\FitnessG}(\SeedSet)\leq\varepsilon$ and
$\max_{\SeedSet}\FP_{\FitnessG}(\SeedSet)> 1-\varepsilon$ in the Heterogeneous Moran process.

First, assume that $\SeedSet$ does not form a set cover, thus \cref{item:hardness_main_upper} of \cref{lem:hardness_main} applies.
We solve the corresponding inequality to arrive at a suitable value for $y$.
In particular, for $\FP_{\FitnessG}(\SeedSet)\leq \varepsilon$, it suffices to find a $y$ small enough such that
\begin{linenomath*}
\begin{align*}
&~1-\left(\frac{\frac{1}{n}}{\frac{1}{n}+(n-1)y}\right)^{n} \leq \varepsilon\\
\Leftrightarrow  &~1-\varepsilon\leq\left(\frac{1}{1+n(n-1)y}\right)^{n} \numberthis \label{eq:one_minus_eps_bound} 
\end{align*}
\end{linenomath*}
Using \cref{lem:prob_bound} for $p=1-\varepsilon$, we set
\begin{linenomath*}
\begin{align*}
&~n(n-1)y =\beta \leq \frac{\ln\left(\frac{1}{p}\right)}{n}\\
\Rightarrow&~y\leq \frac{1}{O(n^3)}
\end{align*}
\end{linenomath*}
as $\varepsilon$ is fixed.
Hence $y$ suffices to be polynomially small in $n$ for the inequality of \cref{item:hardness_main_upper} of  \cref{lem:hardness_main} to hold.

On the other hand, if $\SeedSet$ forms a set cover, \cref{item:hardness_main_lower} of \cref{lem:hardness_main} applies.
We solve the corresponding inequality to arrive at a suitable value for $x$.
In particular, for $\FP_{\FitnessG}(\SeedSet) >1-\varepsilon$, it suffices to find an $x$ large enough such that
\begin{linenomath*}
\begin{align*}
&~\left(\frac{\frac{y}{n^2}\left(\frac{\frac{x}{n}}{\frac{x}{n}+n}\right)^{n}}{1-\left(1-\frac{y}{n^2}\right)\left(\frac{\frac{x}{n}}{\frac{x}{n}+n}\right)^{n}}\right)^n >1-\varepsilon
\end{align*}
\end{linenomath*}
Substituting with $\alpha=\left(\frac{x}{x+n^2}\right)^{n}$, we have
\begin{linenomath*}
\begin{align*}
&~\left(\frac{\frac{y}{n^2}\alpha}{1-\left(1-\frac{y}{n^2}\right)\alpha}\right)^{n} >1-\varepsilon\\
\Rightarrow&~\left( \frac{\frac{y}{n^2}\alpha}{1-\alpha +\frac{y}{n^2}\alpha } \right)^n >1-\varepsilon\\
\Rightarrow&~\left( \frac{1}{1+ \frac{n^2(1-\alpha)}{y\alpha}} \right)^n > 1-\varepsilon
\end{align*}
\end{linenomath*}
Using \cref{lem:prob_bound} for $p=1-\varepsilon$, we set
\begin{linenomath*}
\begin{align*}
&~\frac{n^2(1-\alpha)}{y\alpha}=\beta \leq \frac{\ln\left(\frac{1}{p}\right)}{n}\\
\Rightarrow&~\alpha\geq \frac{n^3}{n^3+yc}
\end{align*}
\end{linenomath*}
where $c=\ln(\sfrac{1}{p})$ is a constant.
We thus have
\begin{linenomath*}
\begin{align*}
&\left(\frac{x}{x+n^2}\right)^{n} \geq \frac{n^3}{n^3+yc}\\
\Rightarrow &~ \left(\frac{1}{1+\frac{n^2}{x}}\right)^n \geq \frac{n^3}{n^3+yc}
\end{align*}
\end{linenomath*}
Using \cref{lem:prob_bound} for $p=\frac{n^3}{n^3+yc}$, we set
\begin{linenomath*}
\begin{align*}
&~\frac{n^2}{x} = \beta \leq \frac{\ln\left(\frac{1}{q}\right)}{n}\\
\Rightarrow&~x\geq \frac{n^3}{\ln\left(\frac{1}{q}\right)} = \frac{n^3}{\ln\left(1+\frac{1}{\frac{n^3}{cy}}\right)}
\end{align*}
\end{linenomath*}
Using \cref{lem:log_lower_bound} for $\zeta=\frac{n^3}{cy}=O(n^6)$ as $c=O(1)$ and $y=1/O(n^3)$, we have
\begin{linenomath*}
\begin{align*}
x\geq\frac{n^3}{\frac{1}{n^7}}=O(n^{10})
\end{align*}
\end{linenomath*}
Thus $x$ suffices to be polynomially large in $n$ for the inequality of \cref{item:hardness_main_lower} of \cref{lem:hardness_main} to hold.
\end{proof}

\thmhardnessmutantbiased*
\begin{proof}
The proof  is by algebraic manipulation on the inequalities of \cref{lem:hardness_main} for the specific case where $y=1$ and $x\geq y$.
Observe that this makes the corresponding graph mutant-biased.
In particular, we argue that there exists an $x$ with a polynomially-long description (in $n$) for which the inequalities stated in the theorem hold.
In turn, this completes a polynomial reduction from Set Cover to distinguishing between $\max_{\SeedSet}\FP_{\FitnessG}(\SeedSet)\leq 1-1/(n^{2n})$ and $\max_{\SeedSet}\FP_{\FitnessG}(\SeedSet)> 1-1/(n^{2n})$ in the Heterogeneous Moran process.

First, assume that $\SeedSet$ does not form a set cover.
\cref{item:hardness_main_upper} of \cref{lem:hardness_main} for $y=1$ gives
\begin{linenomath*}
\begin{align*}
\FP_{\FitnessG}(\SeedSet) \leq 1-\left(\frac{1}{1+(n-1)n}\right)^{n}\leq 1-\frac{1}{n^{2n}}
\end{align*}
\end{linenomath*}

On the other hand, if $\SeedSet$ forms a set cover, \cref{item:hardness_main_lower} of \cref{lem:hardness_main} applies.
For $y=1$, we solve the corresponding inequality to arrive at a suitable value for $x$.
In particular, for $\FP_{\FitnessG}(\SeedSet) >1-\frac{1}{n^{2n}}$, it suffices to find an $x$ large enough such that
\begin{linenomath*}
\begin{align*}
\FP_{\FitnessG}(\SeedSet) \geq 
\left(\frac{\frac{1}{n^2}\left(\frac{\frac{x}{n}}{\frac{x}{n}+n}\right)^{n}}{1-\left(1-\frac{1}{n^2}\right)\left(\frac{\frac{x}{n}}{\frac{x}{n}+n}\right)^{n}}\right)^n>1-\frac{1}{n^{2n}}
\end{align*}
\end{linenomath*}
Substituting with $a=\left(\frac{x}{x+n^2}\right)^{n}$, we have
\begin{linenomath*}
\begin{align*}
&~\left(\frac{\frac{1}{n^2}\alpha}{1-\left(1-\frac{1}{n^2}\right)\alpha}\right)^{n} >1-\frac{1}{n^{2n}}\\
\Rightarrow&~\left( \frac{\frac{1}{n^2}\alpha}{1-\alpha +\frac{1}{n^2}\alpha } \right)^n >1-\frac{1}{n^{2n}}\\
\Rightarrow&~\left( \frac{1}{1+ \frac{n^2(1-\alpha)}{\alpha}} \right)^n > 1-\frac{1}{n^{2n}}
\end{align*}
\end{linenomath*}

Using \cref{lem:prob_bound} for $p=1-\frac{1}{n^{2n}}$, we set
\begin{linenomath*}
\begin{align*}
&~\frac{n^2(1-\alpha)}{\alpha}=\beta \leq \frac{\ln\left(\frac{1}{p}\right)}{n}
\\
\Rightarrow&~\alpha\geq \frac{n^3}{n^3+c}
\end{align*}
\end{linenomath*}
where $c=\ln(\sfrac{1}{p})$.
We thus have
\begin{linenomath*}
\begin{align*}
&\left(\frac{x}{x+n^2}\right)^{n} \geq \frac{n^3}{n^3+c}\\
\Rightarrow &~ \left(\frac{1}{1+\frac{n^2}{x}}\right)^n \geq \frac{n^3}{n^3+c}
\end{align*}
\end{linenomath*}

Using \cref{lem:prob_bound} for $p=\frac{n^3}{n^3+c}$, we set
\begin{linenomath*}
\begin{align*}
&~\frac{n^2}{x} = \beta \leq \frac{\ln\left(\frac{1}{q}\right)}{n}\\
\Rightarrow&~x\geq \frac{n^3}{\ln\left(\frac{1}{q}\right)} = \frac{n^3}{\ln\left(1+\frac{1}{\frac{n^3}{c}}\right)}
\end{align*}
\end{linenomath*}

Using \cref{lem:log_lower_bound} for $\zeta=\frac{n^3}{c}$, we have

\begin{linenomath*}
\begin{align*}
x\geq\frac{n^3}{\frac{1}{\frac{n^4}{c}}}=\frac{n^7}{c}
\end{align*}
\end{linenomath*}

Finally, recall that $c=\ln(\sfrac{1}{p})=\ln\left(1+\frac{1}{n^{2n}-1}\right)$.
Using \cref{lem:log_lower_bound} again with $\beta=n^{2n}-1$, we conclude that

\begin{linenomath*}
\begin{align*}
x\geq\frac{n^7}{\ln\left(1+\frac{1}{n^{2n}-1}\right)}\geq\frac{n^7}{\frac{1}{n^{2n}}}=2^{O(n\log n)}
\end{align*}
\end{linenomath*}
Observe that $x$ has polynomially-long (in $n$) description, thus our reduction from Set Cover is in polynomial time.
\end{proof}

\lemloopyeq*
\begin{proof}
Consider the Heterogeneous and Loopy Moran process, and assume that they are in the same configuration $\ConfigZ\subseteq V$.
Consider any edge $(u,v)\in E$ with $u\neq v$.
Let $p_{u\rightarrow v}$ and $p'_{u\rightarrow v}$ be the probabilities of $u$ transferring its trait to $v$ under the Heterogeneous and Loopy Moran processes, respectively, when in configuration $\ConfigZ$. 
By the definition of the models, we have
\begin{align*}
    p_{u\rightarrow v} = \frac{\Fitness_{\ConfigZ}(u)}{{\TotalFitness}}w(u,v), \,\,\, p'_{u\rightarrow v} &= \frac{1}{n}w_{\ConfigZ}\\(u,v)&=\frac{\Fitness_{\ConfigZ}(u)\cdot\Weight(u,v)}{n\cdot \Fitness_{\max}}.\nonumber
\end{align*}
Moreover, let $R=E\setminus \{(x,x)\colon x\in V\}$ be the set of edges without the self loops in $G$.
Note that, form $\ConfigZ$, each process can progress to a distinct configuration $\ConfigZ'\neq \ConfigZ$ only if a node $u$ transfers its trait along an edge $(u,v)\in R$.
Let $p_1$ and $p'_1$ be the probability that this occurs in the Heterogeneous and Loopy process, repsectively, and we have
\begin{linenomath*}
\begin{align*}
p_1=\sum\limits_{(x,y)\in R}\frac{\Fitness_{\ConfigZ}(x)}{\TotalFitness}w(x,y), \,\,\,\,\,p_1'=&\sum\limits_{(x,y)\in R}\frac{1}{n}w_{\ConfigZ}(x,y)\\=&\sum\limits_{(x,y)\in R}\frac{\Fitness_{\ConfigZ}(x)\cdot\Weight(x,y)}{n\cdot \Fitness_{\max}}. 
\end{align*}
\end{linenomath*}

Finally, observe that
\begin{linenomath*}
\begin{align*}
\frac{p_{u\rightarrow v}}{p_1}=\frac{\frac{\Fitness_{\ConfigZ}(u)}{{\TotalFitness}}w(u,v)}{\sum\limits_{(l,r)\in R}\frac{\Fitness_{\ConfigZ}(l)}{\TotalFitness}w(l,r)}=\frac{{\Fitness_{\ConfigZ}(u)}w(u,v)}{\sum\limits_{(l,r)\in R}{\Fitness_{\ConfigZ}(l)}w(l,r)}\\
\frac{p'_{u\rightarrow v}}{p'_1}=\frac{\frac{\Fitness_{\ConfigZ}(u)\cdot\Weight(u,v)}{n\cdot \Fitness_{\max}}}{\sum\limits_{(l,r)\in R}\frac{\Fitness_{\ConfigZ}(l)\cdot\Weight(l,r)}{n\cdot \Fitness_{\max}} }=\frac{{\Fitness_{\ConfigZ}(u)}w(u,v)}{\sum\limits_{(l,r)\in R}{\Fitness_{\ConfigZ}(l)}w(l,r)}.
\end{align*}
\end{linenomath*}

Thus, the probability distribution $\Prob[\RandomConfig_{t+1} = \Config|\RandomConfig_{t+1} \neq \RandomConfig_{t}]$ is the same in the two processes, yielding the same fixation probability starting from the same seed set, as desired.
\end{proof}
\end{document}